\documentclass[10pt]{article}

\usepackage[english]{babel} 
\usepackage{amsmath}                
\usepackage{amsfonts}               
\usepackage{amssymb}                
\usepackage{amsopn}                 
\allowdisplaybreaks[3] 
\usepackage{amsthm}                
\usepackage{bbm}                    
\usepackage{mathrsfs}               
\usepackage[retainorgcmds]{IEEEtrantools} 
\usepackage{calc}                   
\usepackage[dvips]{graphicx}        
\usepackage{epsfig}                
\usepackage{psfrag}                 
\usepackage[dvips]{color}           
\usepackage{fancyhdr}               
\usepackage{verbatim}               
\usepackage{exscale}                
\usepackage{macros}
\newtheorem{theorem}{Theorem}

\newtheorem{proposition}[theorem]{Proposition}

\title{Increased Capacity per Unit-Cost by Oversampling}

\author{Tobias Koch\\\small University of Cambridge\\\small Cambridge CB2 1PZ, UK\\\small Email: tobi.koch@eng.cam.ac.uk \and Amos Lapidoth\\\small ETH Zurich\\\small 8092 Zurich, Switzerland\\\small
    Email: lapidoth@isi.ee.ethz.ch
 }

\date{}

\begin{document}

\maketitle

\begin{abstract}
It is demonstrated that doubling the sampling rate recovers some of the loss in capacity incurred on the bandlimited Gaussian channel with a one-bit output quantizer.
\renewcommand{\thefootnote}{}
\footnote{The research leading to these results has received funding from the European Community's Seventh Framework Programme (FP7/2007-2013) under grant agreement No. 252663.}
\end{abstract}
\setcounter{footnote}{0}

\section{Introduction}
\label{sec:intro}
We study the capacity of the continuous-time, bandlimited, additive white Gaussian noise (AWGN) channel with one-bit output quantization. Our focus is on the capacity at low transmit powers, i.e., on the capacity per unit-cost, which is defined as the slope of the capacity-vs-input-power curve at zero. We show that increasing the sampling rate reduces the loss in capacity per unit-cost caused by the quantization.

The capacity of the continuous-time AWGN channel without output quantization was studied by Shannon \cite{shannon48}. He showed that if the channel input is bandlimited to $\WW$ Hz and satisfies the average-power constraint $\const{P}$, and if the additive Gaussian noise is of double-sided power spectral density $\Nzero/2$, then the capacity (in nats per second) is given by (see also \cite{gallager68})
\begin{equation}
C(\const{P}) = \WW \log\biggl(1+\frac{\const{P}}{\WW\Nzero}\biggr)
\end{equation}
where $\log(\cdot)$ denotes the natural logarithm function. This capacity can be achieved by transmitting
\begin{equation}
 X(t) = \sum_{\ell=-\infty}^{\infty} X_{\ell} \sinc(2\WW \,t-\ell), \quad t\in\Reals
\end{equation}
(where $\Reals$ denotes the set of real numbers), and by sampling the output $Y(\cdot)$ at Nyquist rate $2\WW$. Here $\{X_{\ell},\,\ell\in\Integers\}$ (where $\Integers$ denotes the set of integers) is a sequence of independent and identically distributed (IID) Gaussian random variables of zero mean and variance $\const{P}$, and $t\mapsto\sinc(t)$ denotes the sinc-function, i.e.,
\begin{equation*}
\sinc(t) = \left\{\begin{array}{ll} 1, \quad & t=0\\ \displaystyle \frac{\sin(\pi t)}{\pi t}, \quad & t \neq 0.\end{array}\right.
\end{equation*}

The above (capacity-achieving) transmission scheme reduces the continuous-time channel to a discrete-time AWGN channel with inputs $\{X_{\ell},\,\ell\in\Integers\}$ and outputs $\bigl\{Y\bigl(\ell/(2\WW)\bigr),\,\ell\in\Integers\bigr\}$. Yet, it is often required that the channel inputs and outputs be not only discrete in time, but also take on a discrete value, i.e., take value in a finite set rather than in $\Reals$. This is, for example, the case if the transmitter and receiver use digital signal processing techniques. To ensure that the channel inputs are discrete-valued, we can simply restrict ourselves to finite input alphabets. This restriction is not critical for small input powers $\const{P}$. Indeed, it is well-known that binary inputs achieve the capacity per unit-cost of the AWGN channel \cite{shannon48}. To ensure that the channel outputs are discrete-valued, we have to employ a quantizer (analog-to-digital converter), which approximates the continuous-valued output by a finite number of bits.

The capacity (in nats per channel use) of the discrete-time AWGN channel with binary symmetric output quantization---where the quantizer produces $1$ for a nonnegative output and $-1$ for a negative output---is given by
\begin{equation}
\label{eq:capacityHL}
\log 2-H_b\Bigl(Q\bigl(\sqrt{\const{P}/\sigma^2}\bigr)\Bigr)
\end{equation}
where $\sigma^2$ denotes the variance of the additive noise, $H_b(\cdot)$ the binary entropy function, and $Q(\cdot)$ the $Q$-function; see \cite[(3.4.18)]{viterbiomura79}, \cite[p.~107]{mceliece02}, \cite[Thm.~2]{singhdabeermadhow09_1}. To the best of our knowledge, there exists no closed-form expression for the capacity of the discrete-time AWGN channel with nonbinary output quantization. However, numerical results are, for example, given in \cite{singhdabeermadhow09_1}. Furthermore, there exist analytical results concerning the capacity per-unit cost. For example, it was demonstrated that if a binary symmetric quantizer is employed, then the capacity per unit-cost equals $\frac{1}{\pi}\frac{1}{\sigma^2}$ \cite[(3.4.20)]{viterbiomura79}. It was further demonstrated that for an octal quantizer with uniform quantization the capacity per unit-cost is not less than $0.475\frac{1}{\sigma^2}$ \cite[(3.4.21)]{viterbiomura79}. Thus, at low transmit power, employing a binary quantizer causes a loss of a factor of $2/\pi$ relative to the capacity per unit-cost $\frac{1}{2}\frac{1}{\sigma^2}$ for unquantized decoding \cite{shannon48}. In contrast, by quantizing the output with 3 bits, a capacity per unit-cost can be achieved that is close to the capacity per unit-cost for unquantized decoding. (Note that the capacity of discrete-time channels is measured in nats per channel use, whereas the capacity of continuous-time channels is measured in nats per second. Since with a continuous-time signal of bandwidth $\WW$ Hz we can approximately transmit $2\WW$ samples per second, we have that one nat per channel use corresponds to $2\WW$ nats per second. By further noting that lowpass filtering and sampling Gaussian noise of double-sided power spectral density $\Nzero/2$ yields Gaussian noise-samples of variance $\WW\Nzero$, it follows that the capacity per unit-cost of the continuous-time channel corresponds to $2\WW$ times the capacity per unit-cost of the discrete-time channel with $\sigma^2$ replaced by $\WW\Nzero$.)

The above results suggest that, in order to reduce the loss in capacity per unit-cost caused by the quantization, one needs to increase the quantizer's resolution. However, while this clearly holds for the discrete-time channel, this does not necessarily hold for the underlying continuous-time channel. Indeed, in contrast to the unquantized channel output, the quantized output is not bandlimited, and it is therefore \emph{prima facie} not clear, whether sampling the quantized output at Nyquist rate is optimal. One might thus increase the capacity of the continuous-time channel by oversampling the quantized output, i.e., by sampling the quantized output at rates higher than the Nyquist rate.

When there is no additive noise, it was shown by Gilbert \cite{gilbert93} and by Shamai \cite{shamai94} that oversampling indeed increases the capacity. In this paper, we demonstrate that oversampling also increases the capacity when the noise power is large relative to the transmit power. In particular, we show that, for binary symmetric output quantization, sampling the quantized output at twice the Nyquist rate yields a capacity per unit-cost that is not less than $0.75\frac{1}{\Nzero}$, which is strictly larger than the capacity per unit-cost $\frac{2}{\pi}\frac{1}{\Nzero}\approx0.64\frac{1}{\Nzero}$ that can be achieved by sampling the quantized output at Nyquist rate.

The rest of this paper is organized as follows. Section~\ref{sec:channel} describes the channel model. Section~\ref{sec:capacity} defines channel capacity and capacity per unit-cost and presents the main result. Section~\ref{sec:proof} provides the proofs of the main result. Section~\ref{sec:summary} concludes the paper with a discussion of our results.

\section{Channel Model}
\label{sec:channel}
\begin{figure}[t]
\centering
\psfrag{x(t)}[r][r]{$x(t)$}
 \psfrag{Z(t)}[cc][cc]{$Z(t)$}
\psfrag{LPF}[cc][cc]{$\textnormal{LPF}_{\WW}$}
\psfrag{k}[tc][tc]{$k\Ts$}
\psfrag{lowpass}[cc][cc]{\footnotesize lowpass filter}
\psfrag{hard-limiter}[cc][cc]{\footnotesize hard-limiter}
\psfrag{Y(k)}[l][l]{$Y(k\Ts)$}
\epsfig{file=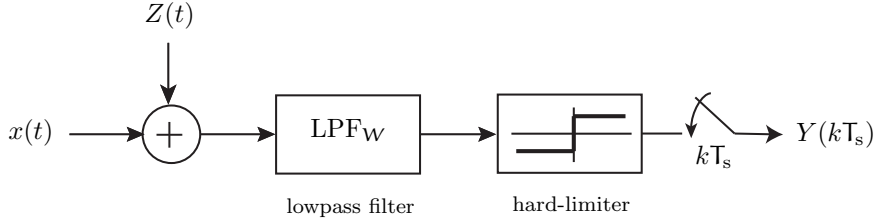, width=0.8\textwidth}
 \caption{System model.}
 \label{fig1}
\end{figure}

We consider the communication channel depicted in Figure~\ref{fig1} whose input $x(\cdot)$ is bandlimited to $\WW$ Hz and satisfies the average-power constraint $\const{P}$. The channel output $Y(k\Ts)$ at integer multiples of the sampling interval $\Ts>0$ is
\begin{equation}
\label{eq:channel}
Y(k\Ts) = \sgn\left\{\bigl(\conv{(\bfx+\bfZ)}{\textnormal{LPF}_{\WW}}\bigr)(k\Ts)\right\}, \quad k\in\Integers
\end{equation}
where $\sgn\{\cdot\}$ denotes the sign function; $(\conv{\vect{a}}{\vect{b}})(t)$ the convolution between $a(\cdot)$ and $b(\cdot)$ at time $t$; and $\textnormal{LPF}_{\WW}(\cdot)$ is the impulse response of the ideal unit-gain lowpass filter of cutoff frequency $\WW$. The hard-limiter is a binary symmetric quantizer that produces $1$ for a nonnegative input and $-1$ for a negative input.  We assume that $\{Z(t),\,t\in\Reals\}$ is white Gaussian noise of double-sided power spectral density $\Nzero/2$. 

Without loss of optimality, we restrict ourselves to signals $x(\cdot)$ of the form
\begin{equation}
\label{eq:PAM}
x(t) = \frac{1}{\sqrt{2\WW}}\sum_{\ell=-\infty}^{\infty} x_{\ell}g\left(t-\frac{\ell}{2\WW}\right), \quad t\in\Reals
\end{equation}
where $g(\cdot)$ is some unit-energy waveform that is bandlimited to $\WW$ Hz. Indeed, by the Sampling Theorem \cite[Thm.~8.4.5]{lapidoth09}, any signal $x(\cdot)$ that is bandlimited to $\WW$ Hz can be written as \eqref{eq:PAM} with
\begin{equation*}
x_{\ell}=x\left(\frac{\ell}{2\WW}\right), \quad \ell\in\Integers \qquad \textnormal{and} \qquad g\left(t\right)=\sqrt{2\WW}\sinc\bigl(2\WW\,t\bigr), \quad t\in\Reals.
\end{equation*}

\section{Channel Capacity and Capacity per Unit-Cost}
\label{sec:capacity}
We define the capacity (in nats per second) as
\begin{equation}
\label{eq:capacity}
C_{\Ts}(\const{P}) \triangleq \varliminf_{n\to\infty} \sup \frac{2\WW}{n} I\bigl(X_1^n;\vect{Y}_1^n\bigr)
\end{equation}
where the supremum is over all unit-energy waveforms $g(\cdot)$ that are bandlimited to $\WW$ Hz and over all joint distributions on $(X_1,X_2,\ldots,X_n)$ satisfying $\frac{1}{n}\sum_{k=1}^n \E{X^2_{k}} \leq \const{P}$. Here $\varliminf$ denotes the \emph{limit inferior}; $A_m^n$ is used to denote the sequence $A_m,A_{m+1},\ldots,A_n$; and
\begin{equation*}
\vect{Y}_k\triangleq\Biggl(Y\left(\left\lceil \frac{2k-1}{4\WW\Ts}\right\rceil\Ts\right),Y\left(\left\lceil \frac{2k-1}{4\WW\Ts}\right\rceil\Ts+\Ts\right),\ldots,Y\left(\left\lfloor \frac{2k+1}{4\WW\Ts}\right\rfloor\Ts\right)\Biggr)
\end{equation*}
(with $\lceil\cdot\rceil$ and $\lfloor\cdot\rfloor$ denoting the ceiling and the floor function). A more general definition of channel capacity for continuous-time channels can be found in \cite[Sec.~8.1]{gallager68}. For the above channel \eqref{eq:channel}, the capacity $C_{\Ts}(\const{P})$ defined by \eqref{eq:capacity} is a lower bound on the capacity defined in \cite[Sec.~8.1]{gallager68}. The two capacities coincide, for example, for the continuous-time AWGN channel (without output quantization).

That oversampling can increase the capacity of the above channel has been demonstrated in the noiseless case, i.e., when $\Nzero=0$. In particular, Gilbert  \cite{gilbert93} showed that, for a Gaussian input $X(\cdot)$, sampling the output at twice the Nyquist rate yields an information rate of $2.14\WW$ bits per second, which is strictly larger than the $2\WW$ bits per second that can be achieved by sampling the output at Nyquist rate. Shamai \cite{shamai94} further showed, \emph{inter alia}, that by sampling the output at $\eta$-times the Nyquist rate, rates of $2\WW\log(1+\eta)$ nats per second are achievable by transmitting a bandlimited process that possesses a single real zero within each Nyquist interval. In the absence of noise it is thus possible to trade amplitude resolution versus time resolution.

In this paper, we focus on the case where the variance of the additive noise is large relative to the transmit power. In particular, we study the capacity per unit-cost, defined as
\begin{equation}
\label{eq:capacityperunitcost}
\dot{C}_{\Ts}(0) \triangleq \varlimsup_{\const{P}\downarrow 0} \frac{C_{\Ts}(\const{P})}{\const{P}}
\end{equation}
(where $\varlimsup$ denotes the \emph{limit superior}). By the Data Processing Inequality \cite[Thm.~2.8.1]{coverthomas91} it follows that quantizing the output does not increase the capacity. This implies that the capacity per unit-cost is upper bounded by the capacity per unit cost of the continuous-time AWGN channel (without output quantization)
\begin{equation}
\dot{C}_{\Ts}(0) \leq \lim_{\const{P}\to\infty}\frac{\WW\log\left(1+\frac{\const{P}}{\WW\Nzero}\right)}{\const{P}} = \frac{1}{\Nzero}.
\end{equation}
For the case where the output is sampled at Nyquist rate $1/\Ts=2\WW$, it was shown that the capacity per unit-cost is given by \cite[(3.4.20)]{viterbiomura79}
\begin{equation}
\dot{C}_{\frac{1}{2\WW}}(0) = \frac{2}{\pi}\frac{1}{\Nzero} \approx 0.637\frac{1}{\Nzero}.
\end{equation}
Thus, when we sample the output at Nyquist rate, hard-limiting causes a loss of a factor of $2/\pi$.  This loss can be reduced by sampling the output at twice the Nyquist rate:

\begin{theorem}[Main Result]
\label{thm:main}
Sampling the output at rate $4\WW$ yields
\begin{IEEEeqnarray}{lCl}
\dot{C}_{\frac{1}{4\WW}}(0) & \geq & \frac{2}{\pi}\frac{1}{\Nzero} \left[\frac{\left(\frac{g_{1}}{2}+\frac{g_0}{4}+\frac{g_{1}}{\pi}\arcsin\left(\frac{\rho}{\sqrt{1-\rho^2}}\right)-\frac{g_0}{2\pi}\arcsin\left(\frac{\rho^2}{1-\rho^2}\right)\right)^2}{\frac{1}{4}+\frac{1}{\pi}\arcsin\left(\rho\right)}\right.\nonumber\\
& & \qquad\qquad {}  + 8 \left(\frac{g_0}{4}-\frac{g_{1}}{\pi}\arcsin\left(\frac{\rho}{\sqrt{1-\rho^2}}\right)+\frac{g_0}{2\pi}\arcsin\left(\frac{\rho^2}{1-\rho^2}\right)\right)^2\nonumber\\
&  & \qquad\qquad {} + \left.\frac{\left(\frac{g_{1}}{2}-\frac{g_0}{4}-\frac{g_{1}}{\pi}\arcsin\left(\frac{\rho}{\sqrt{1-\rho^2}}\right)+\frac{g_0}{2\pi}\arcsin\left(\frac{\rho^2}{1-\rho^2}\right)\right)^2}{\frac{1}{4}-\frac{1}{\pi}\arcsin\left(\rho\right)}\qquad\!\!\right]\nonumber\\
& \approx & 0.747\frac{1}{\Nzero}\IEEEeqnarraynumspace
\end{IEEEeqnarray}
where
\begin{equation*}
\rho=\frac{2}{\pi}, \qquad g_0 = \left.\frac{1+\frac{2}{\pi}\lambda}{\sqrt{\frac{1}{2}\lambda^2+\frac{4}{\pi}\lambda+1}}\right|_{\lambda=1.4},\qquad\textit{and}\qquad g_1 =  \left.\frac{\frac{2}{\pi}+\frac{1}{2}\lambda}{\sqrt{\frac{1}{2}\lambda^2+\frac{4}{\pi}\lambda+1}}\right|_{\lambda=1.4}.
\end{equation*}
\end{theorem}
\begin{proof}
See Section~\ref{sub:mainproof}.
\end{proof}

The main ingredients in the proof of Theorem~\ref{thm:main} are expansions of the complementary cumulative distribution function (CCDF) of bivariate and trivariate Gaussian vectors around the orthant probability.\footnote{The orthant probability is the probability that all components of a random vector have the same sign.} We present these expansions in the following two propositions.

\begin{proposition}
\label{prop:binary}
Let $(x,y)\mapsto\phi_{\vect{0},\mat{K}}(x,y)$ denote the probability density function (PDF) of the bivariate, zero-mean, Gaussian vector of covariance matrix
\begin{equation*}
\mat{K} = \left(\begin{array}{cc} 1 & \varrho\\ \varrho & 1\end{array}\right)
\end{equation*}
for $|\varrho|<1$. Then, for every $\const{A}\geq 0$, $\alpha\in\Reals$ and $\beta\in\Reals$,
\begin{equation}
\int_{-\alpha\const{A}}^{\infty}\int_{-\beta\const{A}}^{\infty} \phi_{\vect{0},\mat{K}}(x,y)\d y\d x =\frac{1}{4}+\frac{1}{2\pi}\arcsin(\varrho)+\frac{\alpha+\beta}{2}\frac{\const{A}}{\sqrt{2\pi}} + \Delta(\const{A},\alpha,\beta)
\end{equation}
where
\begin{equation*}
|\Delta(\const{A},\alpha,\beta)| \leq \const{A}^2 \eta(\const{A},\alpha,\beta)
\end{equation*}
and where $\eta(\const{A},\alpha,\beta)=\eta(\const{A},|\alpha|,|\beta|)$ is monotonically increasing in $(\const{A}, |\alpha|, |\beta|)$ and is bounded for every finite $\const{A}$, $\alpha$, and $\beta$.
\end{proposition}
\begin{proof}
See Section~\ref{subsub:binary}.
\end{proof}

\begin{proposition}
\label{prop:ternary}
Let $(x,y,z)\mapsto \phi_{\vect{0},\mat{K}}(x,y,z)$ denote the PDF of the trivariate, zero-mean, Gaussian vector of covariance matrix
\begin{equation*}
\mat{K} = \left(\begin{array}{ccc} 1 & \varrho_{12} & \varrho_{13} \\ \varrho_{12} & 1 & \varrho_{23} \\ \varrho_{13} & \varrho_{23} & 1\end{array}\right)
\end{equation*}
for $|\varrho_{12}|<1$, $|\varrho_{13}|<1$, $|\varrho_{23}|<1$ satisfying $\det(\mat{K})>0$ (where $\det(\mat{K})$ denotes the determinant of the matrix $\mat{K}$). Then, for every $\const{A}\geq 0$, $\alpha\in\Reals$, $\beta\in\Reals$, and $\gamma\in\Reals$,
\begin{IEEEeqnarray}{lCl}
\IEEEeqnarraymulticol{3}{l}{\int_{-\alpha\const{A}}^{\infty}\int_{-\beta\const{A}}^{\infty}\int_{-\gamma\const{A}}^{\infty} \phi_{\vect{0},\mat{K}}(x,y,z)\d z\d y \d x}\nonumber\\
\,\,\, & = & \frac{1}{8} + \frac{1}{4\pi}\bigl(\arcsin(\varrho_{12})+\arcsin(\varrho_{13})+\arcsin(\varrho_{23})\bigr)\nonumber\\
	& & {} + \frac{\const{A}}{\sqrt{2\pi}}\Biggl[\frac{\alpha+\beta+\gamma}{4}+\frac{\alpha}{2\pi}\arcsin\left(\frac{\varrho_{23}-\varrho_{12}\varrho_{13}}{\sqrt{(1-\varrho_{12}^2)(1-\varrho_{13}^2)}}\right)\nonumber\\
& & \phantom{{}+\frac{\const{A}}{\sqrt{2\pi}}\Biggl[} +\frac{\beta}{2\pi}\arcsin\left(\frac{\varrho_{13}-\varrho_{12}\varrho_{23}}{\sqrt{(1-\varrho_{12}^2)(1-\varrho_{23}^2)}}\right) + \frac{\gamma}{2\pi}\arcsin\left(\frac{\varrho_{12}-\varrho_{13}\varrho_{23}}{\sqrt{(1-\varrho_{13}^2)(1-\varrho_{23}^2)}}\right)\Biggr]\IEEEeqnarraynumspace\nonumber\\
& & {}  + \Delta(\const{A},\alpha,\beta,\gamma)
\end{IEEEeqnarray}
where
\begin{equation*}
|\Delta(\const{A},\alpha,\beta,\gamma)| \leq \const{A}^2 \eta(\const{A},\alpha,\beta,\gamma)
\end{equation*}
and where $\eta(\const{A},\alpha,\beta,\gamma)=\eta(\const{A},|\alpha|,|\beta|,|\gamma|)$ is monotonically increasing in $(\const{A},|\alpha|,|\beta|,|\gamma|)$ and is bounded for every finite $\const{A}$, $\alpha$, $\beta$, and $\gamma$.
\end{proposition}
\begin{proof}
See Section~\ref{subsub:ternary}.
\end{proof}

\section{Proofs}
\label{sec:proof}
The proof of Theorem~\ref{thm:main} is based on the expansions of the CCDF that were presented in Propositions~\ref{prop:binary} and \ref{prop:ternary}. We derive these expansions in Section~\ref{sub:compl}. The proof of Theorem~\ref{thm:main} is given in Section~\ref{sub:mainproof}.

\subsection{Complementary Cumulative Distribution Functions}
\label{sub:compl}
We shall first evaluate the CCDF for the bivariate case. The result will then be used to solve the trivariate case.

\subsubsection{Bivariate Case}
\label{subsub:binary}
In order to prove Proposition~\ref{prop:binary}, we express the CCDF as
\begin{IEEEeqnarray}{lCl}
\IEEEeqnarraymulticol{3}{l}{\int_{-\alpha\const{A}}^{\infty}\int_{-\beta\const{A}}^{\infty} \phi_{\vect{0},\mat{K}}(x,y)\d y\d x}\nonumber\\
\quad & = & \int_{0}^{\infty}\int_0^{\infty} \phi_{\vect{0},\mat{K}}(x,y)\d y\d x + \int_{-\alpha\const{A}}^{0}\int_{-\beta\const{A}}^{\infty} \phi_{\vect{0},\mat{K}}(x,y)\d y\d x + \int_{0}^{\infty}\int_{-\beta\const{A}}^{0} \phi_{\vect{0},\mat{K}}(x,y)\d y\d x \nonumber\\
& = & \frac{1}{4} + \frac{1}{2\pi}\arcsin(\varrho) +  \int_{-\alpha\const{A}}^{0}\int_{-\beta\const{A}}^{\infty} \phi_{\vect{0},\mat{K}}(x,y)\d y\d x + \int_{0}^{\infty}\int_{-\beta\const{A}}^{0} \phi_{\vect{0},\mat{K}}(x,y)\d y\d x\label{eq:binary_1}
\end{IEEEeqnarray}
where the last step follows from the expression for the orthant probability of a bivariate, zero-mean Gaussian vector with correlation coefficient $\varrho$ \cite{sheppard00}, see also \cite{gupta63}. We proceed by evaluating the integrals on the right-hand side (RHS) of \eqref{eq:binary_1}. To evaluate the first integral, we express the joint PDF $(x,y)\mapsto\phi_{\vect{0},\mat{K}}(x,y)$ as the product of a marginal PDF $f(x)$ and a conditional PDF $f(y|x)$, i.e.,
\begin{equation*}
\phi_{\vect{0},\mat{K}}(x,y) = \phi_{0,1}(x)\phi_{\varrho x, 1-\varrho^2}(y), \qquad \bigl(x\in\Reals,\, y\in\Reals\bigr)
\end{equation*}
where $x\mapsto\phi_{\mu,\sigma^2}(x)$ denotes the PDF of a Gaussian random variable of mean $\mu$ and variance $\sigma^2$. We thus obtain for the first integral on the RHS of \eqref{eq:binary_1}
\begin{IEEEeqnarray}{lCl}
\int_{-\alpha\const{A}}^{0}\int_{-\beta\const{A}}^{\infty} \phi_{\vect{0},\mat{K}}(x,y)\d y\d x & = & \int_{-\alpha\const{A}}^0 \phi_{0,1}(x) \int_{-\beta\const{A}}^{\infty} \phi_{\varrho x, 1-\varrho^2}(y) \d y \d x\nonumber\\
& = & \int_{-\alpha\const{A}}^0 \phi_{0,1}(x)\left[1-Q\left(\frac{\beta\const{A}+\varrho x}{\sqrt{1-\varrho^2}}\right)\right]\d x
\end{IEEEeqnarray}
where $Q(\cdot)$ denotes the $Q$-function
\begin{equation*}
Q(x) \triangleq \frac{1}{\sqrt{2\pi}} \int_{x}^{\infty} \exp\left(-\frac{\xi^2}{2}\right)\d\xi, \qquad x\in\Reals.
\end{equation*}
Expressing $x\mapsto Q(x)$ as a Taylor series around zero yields \cite[(3.54)]{verdu98_book}
\begin{IEEEeqnarray}{lCl}
\int_{-\alpha\const{A}}^{0}\int_{-\beta\const{A}}^{\infty} \phi_{\vect{0},\mat{K}}(x,y)\d y\d x & = &  \int_{-\alpha\const{A}}^0 \phi_{0,1}(x) \left[\frac{1}{2}+\frac{1}{\sqrt{2\pi}}\frac{\beta\const{A}+\varrho x}{\sqrt{1-\varrho^2}}-\delta\left(\frac{\beta\const{A}+\varrho x}{\sqrt{1-\varrho^2}}\right)\right]\d x\nonumber\\
& = & \frac{1}{2}\int_{-\alpha\const{A}}^0 \phi_{0,1}(x) \d x + \frac{1}{2}\delta(\alpha\const{A}) + \Delta_1(\const{A},\alpha,\beta)\label{eq:binary_2}
\end{IEEEeqnarray}
where
\begin{IEEEeqnarray}{lCl}
\Delta_1(\const{A},\alpha,\beta) & \triangleq & \frac{\beta\const{A}}{\sqrt{2\pi(1-\varrho^2)}} \int_{-\alpha\const{A}}^0 \phi_{0,1}(x)\d x + \frac{\varrho}{\sqrt{2\pi(1-\varrho^2)}} \int_{-\alpha\const{A}}^0 \phi_{0,1}(x)x\d x\nonumber\\
& & {}  - \int_{-\alpha\const{A}}^0 \phi_{0,1}(x)\delta\left(\frac{\beta\const{A}+\rho x}{\sqrt{1-\varrho^2}}\right)\d x - \frac{1}{2}\delta(\alpha\const{A})\label{eq:binary_Delta1}
\end{IEEEeqnarray}
and where $\delta(\cdot)$ denotes the remainder term of the Taylor series expansion of $Q(\cdot)$, which satisfies \cite[Sec.~0.317]{gradshteynryzhik00}
\begin{equation}
\label{eq:delta_binary}
|\delta(x)| \leq \frac{|x|^3}{6}\frac{1}{\sqrt{2\pi}},\qquad |x|\leq 1.
\end{equation}
Expressing again $x\mapsto Q(x)$ as a Taylor series around zero, the first two terms on the RHS of \eqref{eq:binary_2} are evaluated as
\begin{IEEEeqnarray}{lCl}
\frac{1}{2}\int_{-\alpha\const{A}}^0 \phi_{0,1}(x) \d x + \frac{1}{2}\delta(\alpha\const{A}) & = & \frac{1}{2}\left(\frac{1}{2}-Q(\alpha\const{A})\right)+\frac{1}{2}\delta{\alpha\const{A}}\nonumber\\
& = & \frac{\alpha}{2}\frac{\const{A}}{\sqrt{2\pi}}. \label{eq:binary_Q}
\end{IEEEeqnarray}
It is shown in Appendix~\ref{app:bivariate} that $\Delta_1(\const{A},\alpha,\beta)$ satisfies
\begin{equation*}
|\Delta_1(\const{A},\alpha,\beta)| \leq \const{A}^2 \eta_1(\const{A},\alpha,\beta)
\end{equation*}
where $\eta_1(\const{A},\alpha,\beta)=\eta_1(\const{A},|\alpha|,|\beta|)$ is monotonically increasing in $(\const{A},|\alpha|,|\beta|)$ and is bounded for every finite $\const{A}$, $\alpha$, and $\beta$.

Along the same lines, we evaluate the second integral on the RHS of \eqref{eq:binary_1} as
\begin{IEEEeqnarray}{lCl}
\int_{0}^{\infty}\int_{-\beta\const{A}}^{0} \phi_{\vect{0},\mat{K}}(x,y)\d y\d x & = & \int_{-\beta\const{A}}^0 \phi_{0,1}(y) \int_0^{\infty} \phi_{\varrho y, 1-\varrho^2}(x)\d x \d y \nonumber\\
& = &  \int_{-\beta\const{A}}^0 \phi_{0,1}(y) \left[1-Q\left(\frac{\varrho y}{\sqrt{1-\varrho^2}}\right)\right]\d y\nonumber\\
& = &  \int_{-\beta\const{A}}^0 \phi_{0,1}(y)\left[\frac{1}{2}+\frac{1}{\sqrt{2\pi}}\frac{\varrho y}{\sqrt{1-\varrho^2}}-\delta\left(\frac{\varrho y}{\sqrt{1-\varrho^2}}\right)\right]\d y \nonumber\\
& = & \frac{1}{2}  \int_{-\beta\const{A}}^0 \phi_{0,1}(y) \d y + \frac{1}{2}\delta(\beta\const{A}) + \Delta_2(\const{A},\alpha,\beta)\nonumber\\
& = & \frac{\beta}{2}\frac{\const{A}}{\sqrt{2\pi}} + \Delta_2(\const{A},\alpha,\beta) \label{eq:binary_3}
\end{IEEEeqnarray}
where
\begin{IEEEeqnarray}{lCl}
\Delta_2(\const{A},\alpha,\beta) \triangleq \frac{\varrho}{\sqrt{2\pi(1-\varrho^2)}}\int_{-\beta\const{A}}^0 \phi_{0,1}(y) y \d y - \int_{-\beta\const{A}}^0 \phi_{0,1}(y) \delta\left(\frac{\varrho y}{\sqrt{1-\varrho^2}}\right)\d y -\frac{1}{2}\delta(\beta\const{A}). \IEEEeqnarraynumspace\label{eq:binary_Delta2}
\end{IEEEeqnarray}
It is shown in Appendix~\ref{app:bivariate} that $\Delta_2(\const{A},\alpha,\beta)$ satisfies
\begin{equation*}
|\Delta_2(\const{A},\alpha,\beta)| \leq \const{A}^2\eta_2(\const{A},\alpha,\beta)
\end{equation*}
where $\eta_2(\const{A},\alpha,\beta)=\eta_2(\const{A},|\alpha|,|\beta|)$ is monotonically increasing in $(\const{A},|\alpha|,|\beta|)$ and is bounded for every finite $\const{A}$, $\alpha$, and $\beta$.

Combining \eqref{eq:binary_1}, \eqref{eq:binary_2}, \eqref{eq:binary_Q}, and \eqref{eq:binary_3} yields
\begin{equation}
\int_{-\alpha\const{A}}^{\infty}\int_{-\beta\const{A}}^{\infty} \phi_{\vect{0},\mat{K}}(x,y)\d y\d x = \frac{1}{4} + \frac{1}{2\pi}\arcsin(\varrho) + \frac{\alpha+\beta}{2}\frac{\const{A}}{\sqrt{2\pi}} + \Delta(\const{A},\alpha,\beta) \label{eq:binary_4}
\end{equation}
where
\begin{equation*}
\Delta(\const{A},\alpha,\beta) \triangleq \Delta_1(\const{A},\alpha,\beta)+\Delta_2(\const{A},\alpha,\beta).
\end{equation*}
By the Triangle Inequality \cite[Sec.~2.4]{lapidoth09}, we have
\begin{IEEEeqnarray}{lCl}
|\Delta(\const{A},\alpha,\beta)| & \leq & |\Delta_1(\const{A},\alpha,\beta)| + |\Delta_2(\const{A},\alpha,\beta)|\nonumber\\
& \leq & \const{A}^2\bigl(\eta_1(\const{A},\alpha,\beta)+\eta_2(\const{A},\alpha,\beta)\bigr). \label{eq:binary_5}
\end{IEEEeqnarray}
Proposition~\ref{prop:binary} follows now from \eqref{eq:binary_4} and \eqref{eq:binary_5} by noting that if $\eta_1(\const{A},\alpha,\beta)$ and $\eta_2(\const{A},\alpha,\beta)$ are both monotonically increasing in $(\const{A},|\alpha|,|\beta|)$, then so is 
\begin{equation*}
\eta(\const{A},\alpha,\beta) \triangleq \eta_1(\const{A},\alpha,\beta)+\eta_2(\const{A},\alpha,\beta);
\end{equation*}
and if $\eta_1(\const{A},\alpha,\beta)$ and $\eta_2(\const{A},\alpha,\beta)$ are bounded, then so is $\eta(\const{A},\alpha,\beta)$.

\subsubsection{Trivariate Case}
\label{subsub:ternary}
In order to prove Proposition~\ref{prop:ternary}, we express the CCDF as
\begin{IEEEeqnarray}{lCl}
\IEEEeqnarraymulticol{3}{l}{\int_{-\alpha\const{A}}^{\infty}\int_{-\beta\const{A}}^{\infty}\int_{-\gamma\const{A}}^{\infty}\phi_{\vect{0},\mat{K}}(x,y,z)\d z\d y \d x}\nonumber\\
\quad & = & \int_{0}^{\infty} \int_0^{\infty} \int_0^{\infty} \phi_{\vect{0},\mat{K}}(x,y,z)\d z\d y \d x + \int_{-\alpha\const{A}}^0\int_{-\beta\const{A}}^{\infty}\int_{-\gamma\const{A}}^{\infty} \phi_{\vect{0},\mat{K}}(x,y,z)\d z\d y\d x\nonumber\\
& & {} + \int_{0}^{\infty}\int_{-\beta\const{A}}^0\int_{-\gamma\const{A}}^{\infty} \phi_{\vect{0},\mat{K}}(x,y,z)\d z\d y\d x + \int_0^{\infty}\int_0^{\infty}\int_{-\gamma\const{A}}^0 \phi_{\vect{0},\mat{K}}(x,y,z)\d z\d y\d x\nonumber\\
& = & \frac{1}{8} + \frac{1}{4\pi}\bigl(\arcsin(\varrho_{12})+\arcsin(\varrho_{13})+\arcsin(\varrho_{23})\bigr) + \int_{-\alpha\const{A}}^0\int_{-\beta\const{A}}^{\infty}\int_{-\gamma\const{A}}^{\infty} \phi_{\vect{0},\mat{K}}(x,y,z)\d z\d y\d x\nonumber\\
& & {} + \int_{0}^{\infty}\int_{-\beta\const{A}}^0\int_{-\gamma\const{A}}^{\infty} \phi_{\vect{0},\mat{K}}(x,y,z)\d z\d y\d x + \int_0^{\infty}\int_0^{\infty}\int_{-\gamma\const{A}}^0 \phi_{\vect{0},\mat{K}}(x,y,z)\d z\d y\d x\label{eq:ternary_1}
\end{IEEEeqnarray}
where the last step follows from the expression for the orthant probability of a trivariate, zero-mean Gaussian vector of covariance matrix \cite{david53} (see also \cite{gupta63})
\begin{equation*}
\mat{K} = \left(\begin{array}{ccc} 1 & \varrho_{12} & \varrho_{13} \\ \varrho_{12} & 1 & \varrho_{23} \\ \varrho_{13} & \varrho_{23} & 1\end{array}\right).
\end{equation*}
If $\const{A}=0$, then the integrals on the RHS of \eqref{eq:ternary_1} are zero and Proposition~\ref{prop:ternary} follows directly from \eqref{eq:ternary_1}. To prove Proposition~\ref{prop:ternary} for $\const{A}>0$, we continue by evaluating the integrals on the RHS of \eqref{eq:ternary_1} separately. To evaluate the first integral, we express the joint PDF $(x,y,z)\mapsto\phi_{\vect{0},\mat{K}}(x,y,z)$ as the product of a marginal PDF $f(x)$ and a conditional PDF $f(y,z|x)$
\begin{equation*}
\phi_{\vect{0},\mat{K}}(x,y,z) = \phi_{0,1}(x)\phi_{\bfmu(x),\mat{K}(x)}(y,z), \qquad \bigl(x\in\Reals,\, y\in\Reals,\, z\in\Reals\bigr)
\end{equation*}
where $(y,z)\mapsto\phi_{\bfmu,\mat{K}}(\cdot)(y,z)$ denotes the PDF of a bivariate Gaussian vector of mean $\bfmu$ and covariance matrix $\mat{K}$, and where
\begin{equation*}
\bfmu(x) = \left(\begin{array}{c}\varrho_{12}x \\ \varrho_{13}x\end{array}\right) \qquad \textnormal{and} \qquad \mat{K}(x) = \left(\begin{array}{cc} 1-\varrho_{12}^2 & \varrho_{23}-\varrho_{12}\varrho_{13} \\ \varrho_{23}-\varrho_{12}\varrho_{13} & 1-\varrho_{13}^2\end{array}\right).
\end{equation*}
We thus have
\begin{IEEEeqnarray}{lCl}
\IEEEeqnarraymulticol{3}{l}{\int_{-\alpha\const{A}}^0\int_{-\beta\const{A}}^{\infty}\int_{-\gamma\const{A}}^{\infty} \phi_{\vect{0},\mat{K}}(x,y,z)\d z\d y\d x}\nonumber\\
\,\, & = & \int_{-\alpha\const{A}}^0\phi_{0,1}(x)\int_{-\beta\const{A}}^{\infty}\int_{-\gamma\const{A}}^{\infty} \phi_{\bfmu(x),\mat{K}(x)}(y,z)\d z\d y\d x\nonumber\\
& = & \int_{-\alpha\const{A}}^0\phi_{0,1}(x)\int_{-\frac{\beta\const{A}+\varrho_{12}x}{\sqrt{1-\varrho_{12}^2}}}^{\infty}\int_{-\frac{\gamma\const{A}+\varrho_{13}x}{\sqrt{1-\varrho_{13}^2}}}^{\infty} \phi_{\vect{0},\mat{K}'(x)}(y',z')\d z'\d y'\d x \nonumber\\
& = &  \int_{-\alpha\const{A}}^0\phi_{0,1}(x)\left[\frac{1}{4}+\frac{1}{2\pi}\arcsin\left(\frac{\varrho_{23}-\varrho_{12}\varrho_{13}}{\sqrt{(1-\varrho_{12}^2)(1-\varrho_{13}^2)}}\right)\right] \d x \nonumber\\
& & {} + \int_{-\alpha\const{A}}^0\phi_{0,1}(x) \frac{1}{2\sqrt{2\pi}}\left(\frac{\beta\const{A}+\varrho_{12}x}{\sqrt{1-\varrho_{12}^2}}+\frac{\gamma\const{A}+\varrho_{13}x}{\sqrt{1-\varrho_{13}^2}}\right)\d x \nonumber\\
& & {} + \int_{-\alpha\const{A}}^0\phi_{0,1}(x) \Delta\left(\const{A},\frac{\beta+\varrho_{12}x/\const{A}}{\sqrt{1-\varrho_{12}^2}},\frac{\gamma+\varrho_{13}x/\const{A}}{\sqrt{1-\varrho_{13}^2}}\right)\d x\nonumber\\
& = &  \left(\int_{-\alpha\const{A}}^0\phi_{0,1}(x)\d x+\delta(\alpha\const{A})\right)\left[\frac{1}{4}+\frac{1}{2\pi}\arcsin\left(\frac{\varrho_{23}-\varrho_{12}\varrho_{13}}{\sqrt{(1-\varrho_{12}^2)(1-\varrho_{13}^2)}}\right)\right] + \Delta_1(\const{A},\alpha,\beta,\gamma)\IEEEeqnarraynumspace\label{eq:ternary_2}
\end{IEEEeqnarray}
where $\Delta(\const{A},\alpha,\beta)$ and $\delta(\cdot)$ are as in the previous section, and where
\begin{equation*}
\mat{K}'(x) = \left(\begin{array}{cc} 1 & \frac{\varrho_{23}-\varrho_{12}\varrho_{13}}{\sqrt{(1-\varrho_{12}^2)(1-\varrho_{13}^2)}} \\ \frac{\varrho_{23}-\varrho_{12}\varrho_{13}}{\sqrt{(1-\varrho_{12}^2)(1-\varrho_{13}^2)}} & 1\end{array}\right)
\end{equation*}
and
\begin{IEEEeqnarray}{lCl}
\Delta_1(\const{A},\alpha,\beta,\gamma) & \triangleq & -\delta(\alpha\const{A})\left[\frac{1}{4}+\frac{1}{2\pi}\arcsin\left(\frac{\varrho_{23}-\varrho_{12}\varrho_{13}}{\sqrt{(1-\varrho_{12}^2)(1-\varrho_{13}^2)}}\right)\right] \nonumber\\
& & {} + \int_{-\alpha\const{A}}^0\phi_{0,1}(x) \frac{1}{2\sqrt{2\pi}}\left(\frac{\beta\const{A}+\varrho_{12}x}{\sqrt{1-\varrho_{12}^2}}+\frac{\gamma\const{A}+\varrho_{13}x}{\sqrt{1-\varrho_{13}^2}}\right)\d x \nonumber\\
& & {} + \int_{-\alpha\const{A}}^0\phi_{0,1}(x) \Delta\left(\const{A},\frac{\beta+\varrho_{12}x/\const{A}}{\sqrt{1-\varrho_{12}^2}},\frac{\gamma+\varrho_{13}x/\const{A}}{\sqrt{1-\varrho_{13}^2}}\right)\d x.\label{eq:ternary_Delta1}
\end{IEEEeqnarray}
Here the second step follows by substituting
\begin{equation*}
y'=\frac{y-\varrho_{12}x}{\sqrt{1-\varrho_{12}^2}} \qquad \textnormal{and} \qquad z'=\frac{z-\varrho_{13}x}{\sqrt{1-\varrho_{13}^2}}
\end{equation*}
and the third step follows from Proposition~\ref{prop:binary}.

By applying \eqref{eq:binary_Q}, we obtain for the first term on the RHS of \eqref{eq:ternary_2}
\begin{IEEEeqnarray}{lCl}
\IEEEeqnarraymulticol{3}{l}{ \left(\int_{-\alpha\const{A}}^0\phi_{0,1}(x)\d x+\delta(\alpha\const{A})\right)\left[\frac{1}{4}+\frac{1}{2\pi}\arcsin\left(\frac{\varrho_{23}-\varrho_{12}\varrho_{13}}{\sqrt{(1-\varrho_{12}^2)(1-\varrho_{13}^2)}}\right)\right]}\nonumber\\
\qquad\qquad\qquad\qquad\qquad\qquad\qquad & = & \frac{\const{A}}{\sqrt{2\pi}}\left[\frac{\alpha}{4}+\frac{\alpha}{2\pi}\arcsin\left(\frac{\varrho_{23}-\varrho_{12}\varrho_{13}}{\sqrt{(1-\varrho_{12}^2)(1-\varrho_{13}^2)}}\right)\right]. \label{eq:ternary_3}
\end{IEEEeqnarray}
It is shown in Appendix~\ref{app:trivariate} that $\Delta_1(\const{A},\alpha,\beta,\gamma)$ satisfies
\begin{equation*}
|\Delta_1(\const{A},\alpha,\beta,\gamma)| \leq \const{A}^2 \eta_1(\const{A},\alpha,\beta,\gamma)
\end{equation*}
where $\eta_1(\const{A},\alpha,\beta,\gamma)=\eta_1(\const{A},|\alpha|,|\beta|,|\gamma|)$ is monotonically increasing in $(\const{A},|\alpha|,|\beta|,|\gamma|)$ and is bounded for every finite $\const{A}$, $\alpha$, $\beta$, and $\gamma$.

To evaluate the second integral on the RHS of \eqref{eq:ternary_1}, we express the joint PDF $(x,y,z)\mapsto\phi_{0,\mat{K}}(x,y,z)$ as
\begin{equation*}
\phi_{0,\mat{K}}(x,y,z) = \phi_{0,1}(y) \phi_{\bfmu(y),\mat{K}(y)}(x,z), \qquad \bigl(x\in\Reals,\, y\in\Reals,\, z\in\Reals\bigr)
\end{equation*}
where
\begin{equation*}
\bfmu(y) = \left(\begin{array}{c} \varrho_{12}y \\ \varrho_{23} y\end{array}\right) \qquad \textnormal{and} \qquad \mat{K}(y) = \left(\begin{array}{cc} 1-\varrho_{12}^2 & \varrho_{13}-\varrho_{12}\varrho_{23} \\ \varrho_{13}-\varrho_{12}\varrho_{23} & 1-\varrho_{23}^2\end{array}\right).
\end{equation*}
We thus have
\begin{IEEEeqnarray}{lCl}
\IEEEeqnarraymulticol{3}{l}{\int_{0}^{\infty}\int_{-\beta\const{A}}^0\int_{-\gamma\const{A}}^{\infty} \phi_{\vect{0},\mat{K}}(x,y,z)\d z\d y\d x}\nonumber\\
\quad & = & \int_{-\beta\const{A}}^0\phi_{0,1}(y) \int_0^{\infty}\int_{-\gamma\const{A}}^{\infty} \phi_{\bfmu(y),\mat{K}(y)}(x,z)\d z\d x \d y\nonumber\\
& = & \int_{-\beta\const{A}}^0\phi_{0,1}(y) \int_{-\frac{\varrho_{12}y}{\sqrt{1-\varrho_{12}^2}}}^{\infty}\int_{-\frac{\gamma\const{A}+\varrho_{23}y}{\sqrt{1-\varrho_{23}^2}}}^{\infty} \phi_{\vect{0},\mat{K}'(y)}(x',z')\d z'\d x' \d y\nonumber\\
& = & \int_{-\beta\const{A}}^0\phi_{0,1}(y) \left[\frac{1}{4} + \frac{1}{2\pi}\arcsin\left(\frac{\varrho_{13}-\varrho_{12}\varrho_{23}}{\sqrt{(1-\varrho_{12}^2)(1-\varrho_{23}^2)}}\right)\right] \d y\nonumber\\
& & {} + \int_{-\beta\const{A}}^0\phi_{0,1}(y) \frac{1}{2\sqrt{2\pi}}\left(\frac{\varrho_{12}y}{\sqrt{1-\varrho_{12}^2}}+\frac{\gamma\const{A}+\varrho_{23}y}{\sqrt{1-\varrho_{23}^2}}\right)\d y\nonumber\\
& & {} + \int_{-\beta\const{A}}^0\phi_{0,1}(y) \Delta\left(\const{A},\frac{\varrho_{12}y/\const{A}}{\sqrt{1-\varrho_{12}^2}},\frac{\gamma+\varrho_{13}y/\const{A}}{\sqrt{1-\varrho_{23}^2}}\right)\d y\nonumber\\
& = &  \left(\int_{-\beta\const{A}}^0\phi_{0,1}(y)\d y + \delta(\beta\const{A}) \right) \left[\frac{1}{4} + \frac{1}{2\pi}\arcsin\left(\frac{\varrho_{13}-\varrho_{12}\varrho_{23}}{\sqrt{(1-\varrho_{12}^2)(1-\varrho_{23}^2)}}\right)\right]  + \Delta_2(\const{A},\alpha,\beta,\gamma)\nonumber\\
& = & \frac{\const{A}}{\sqrt{2\pi}}\left[\frac{\beta}{4} + \frac{\beta}{2\pi}\arcsin\left(\frac{\varrho_{13}-\varrho_{12}\varrho_{23}}{\sqrt{(1-\varrho_{12}^2)(1-\varrho_{23}^2)}}\right)\right] + \Delta_2(\const{A},\alpha,\beta,\gamma) \label{eq:ternary_4}
\end{IEEEeqnarray}
where
\begin{equation*}
\mat{K}'(y) = \left(\begin{array}{cc} 1 & \frac{\varrho_{13}-\varrho_{12}\varrho_{23}}{\sqrt{(1-\varrho_{12}^2)(1-\varrho_{23}^2)}} \\ \frac{\varrho_{13}-\varrho_{12}\varrho_{23}}{\sqrt{(1-\varrho_{12}^2)(1-\varrho_{23}^2)}} & 1 \end{array}\right)
\end{equation*}
and
\begin{IEEEeqnarray}{lCl}
\Delta_2(\const{A},\alpha,\beta,\gamma) & \triangleq & -\delta(\beta\const{A})\left[\frac{1}{4} + \frac{1}{2\pi}\arcsin\left(\frac{\varrho_{13}-\varrho_{12}\varrho_{23}}{\sqrt{(1-\varrho_{12}^2)(1-\varrho_{23}^2)}}\right)\right] \nonumber\\
& & {} +  \int_{-\beta\const{A}}^0\phi_{0,1}(y) \frac{1}{2\sqrt{2\pi}}\left(\frac{\varrho_{12}y}{\sqrt{1-\varrho_{12}^2}}+\frac{\gamma\const{A}+\varrho_{23}y}{\sqrt{1-\varrho_{23}^2}}\right)\d y\nonumber\\
& & {} + \int_{-\beta\const{A}}^0\phi_{0,1}(y) \Delta\left(\const{A},\frac{\varrho_{12}y/\const{A}}{\sqrt{1-\varrho_{12}^2}},\frac{\gamma+\varrho_{23}y/\const{A}}{\sqrt{1-\varrho_{23}^2}}\right)\d y.\label{eq:ternary_Delta2}
\end{IEEEeqnarray}
Here the second step follows by substituting
\begin{equation*}
x'=\frac{x-\varrho_{12}y}{\sqrt{1-\varrho_{12}^2}} \qquad \textnormal{and} \qquad z'=\frac{z-\varrho_{23}y}{\sqrt{1-\varrho_{23}^2}};
\end{equation*}
the third step follows from Proposition~\ref{prop:binary}; and the last step follows from \eqref{eq:binary_Q}. It is shown in Appendix~\ref{app:trivariate} that $\Delta_2(\const{A},\alpha,\beta,\gamma)$ satisfies
\begin{equation*}
|\Delta_2(\const{A},\alpha,\beta,\gamma)| \leq \const{A}^2 \eta_2(\const{A},\alpha,\beta,\gamma)
\end{equation*}
where $\eta_2(\const{A},\alpha,\beta,\gamma)=\eta_2(\const{A},|\alpha|,|\beta|,|\gamma|)$ is monotonically increasing in $(\const{A},|\alpha|,|\beta|,|\gamma|)$ and is bounded for every finite $\const{A}$, $\alpha$, $\beta$, and $\gamma$.

To evaluate the third integral on the RHS of \eqref{eq:ternary_1}, we express the joint PDF $(x,y,z)\mapsto\phi_{0,\mat{K}}(x,y,z)$ as
\begin{equation*}
\phi_{0,\mat{K}}(x,y,z) = \phi_{0,1}(z)\phi_{\bfmu(z),\mat{K}(z)}(x,y), \qquad \bigl(x\in\Reals,\, y\in\Reals,\, z\in\Reals\bigr)
\end{equation*}
where
\begin{equation*}
\bfmu(z) = \left(\begin{array}{c}\varrho_{13}z \\ \varrho_{23}z\end{array}\right) \qquad \textnormal{and} \qquad \mat{K}(z) = \left(\begin{array}{cc} 1-\varrho_{13}^2 & \varrho_{12}-\varrho_{13}\varrho_{23} \\ \varrho_{12}-\varrho_{13}\varrho_{23} & 1\end{array}\right).
\end{equation*}
We have
\begin{IEEEeqnarray}{lCl}
\IEEEeqnarraymulticol{3}{l}{\int_0^{\infty}\int_0^{\infty}\int_{-\gamma\const{A}}^0 \phi_{\vect{0},\mat{K}}(x,y,z)\d z\d y\d x}\nonumber\\
\quad & = & \int_{-\gamma\const{A}}^0 \phi_{0,1}(z) \int_{0}^{\infty} \int_0^{\infty} \phi_{\bfmu(z),\mat{K}(z)}(x,y) \d y\d x \d z \nonumber\\
& = & \int_{-\gamma\const{A}}^0 \phi_{0,1}(z) \int_{-\frac{\varrho_{13}z}{\sqrt{1-\varrho_{13}^2}}}^{\infty} \int_{-\frac{\varrho_{23}z}{\sqrt{1-\varrho_{23}^2}}}^{\infty} \phi_{\vect{0},\mat{K}'(z)}(x',y') \d y'\d x'\d z \nonumber\\
& = & \int_{-\gamma\const{A}}^0 \phi_{0,1}(z)\left[\frac{1}{4}+\frac{1}{2\pi}\arcsin\left(\frac{\varrho_{12}-\varrho_{13}\varrho_{23}}{\sqrt{(1-\varrho_{13}^2)(1-\varrho_{23}^2)}}\right)\right]\d z \nonumber\\
& & {} +  \int_{-\gamma\const{A}}^0 \phi_{0,1}(z) \frac{1}{2\sqrt{2\pi}}\left(\frac{\varrho_{13}z}{\sqrt{1-\varrho_{13}^2}}+\frac{\varrho_{23}z}{\sqrt{1-\varrho_{23}^2}}\right) \d z \nonumber\\
& & {} +  \int_{-\gamma\const{A}}^0 \phi_{0,1}(z)\Delta\left(\const{A},\frac{\varrho_{13}z/\const{A}}{\sqrt{1-\varrho_{13}^2}},\frac{\varrho_{23}z/\const{A}}{\sqrt{1-\varrho_{23}^2}}\right)\d z\nonumber\\
& = &  \left(\int_{-\gamma\const{A}}^0 \phi_{0,1}(z)\d z + \delta(\gamma\const{A})\right) \left[\frac{1}{4}+\frac{1}{2\pi}\arcsin\left(\frac{\varrho_{12}-\varrho_{13}\varrho_{23}}{\sqrt{(1-\varrho_{13}^2)(1-\varrho_{23}^2)}}\right)\right] + \Delta_3(\const{A},\alpha,\beta,\gamma) \nonumber\\
& = & \frac{\const{A}}{\sqrt{2\pi}}\left[\frac{\gamma}{4}+\frac{\gamma}{2\pi}\arcsin\left(\frac{\varrho_{12}-\varrho_{13}\varrho_{23}}{\sqrt{(1-\varrho_{13}^2)(1-\varrho_{23}^2)}}\right)\right] + \Delta_3(\const{A},\alpha,\beta,\gamma) \label{eq:ternary_5}
\end{IEEEeqnarray}
where
\begin{equation*}
\mat{K}'(z) = \left(\begin{array}{cc} 1 & \frac{\varrho_{12}-\varrho_{13}\varrho_{23}}{\sqrt{(1-\varrho_{13}^2)(1-\varrho_{23}^2)}} \\ \frac{\varrho_{12}-\varrho_{13}\varrho_{23}}{\sqrt{(-\varrho_{13}^2)(1-\varrho_{23}^2)}} & 1 \end{array}\right)
\end{equation*}
and
\begin{IEEEeqnarray}{lCl}
\Delta_3(\const{A},\alpha,\beta,\gamma) & \triangleq & -\delta(\gamma\const{A})\left[\frac{1}{4}+\frac{1}{2\pi}\arcsin\left(\frac{\varrho_{12}-\varrho_{13}\varrho_{23}}{\sqrt{(1-\varrho_{13}^2)(1-\varrho_{23}^2)}}\right)\right]\nonumber\\
& & {} +  \int_{-\gamma\const{A}}^0 \phi_{0,1}(z) \frac{1}{2\sqrt{2\pi}}\left(\frac{\varrho_{13}z}{\sqrt{1-\varrho_{13}^2}}+\frac{\varrho_{23}z}{\sqrt{1-\varrho_{23}^2}}\right) \d z \nonumber\\
& & {} +  \int_{-\gamma\const{A}}^0 \phi_{0,1}(z)\Delta\left(\const{A},\frac{\varrho_{13}z/\const{A}}{\sqrt{1-\varrho_{13}^2}},\frac{\varrho_{23}z/\const{A}}{\sqrt{1-\varrho_{23}^2}}\right)\d z.\label{eq:ternary_Delta3}
\end{IEEEeqnarray}
Here the second step follows by substituting
\begin{equation*}
x' = \frac{x-\varrho_{13}z}{\sqrt{1-\varrho_{13}^2}} \qquad \textnormal{and} \qquad y' = \frac{y-\varrho_{23}z}{\sqrt{1-\varrho_{23}^2}};
\end{equation*}
the third step follows from Proposition~\ref{prop:binary}; and the last step follows from \eqref{eq:binary_Q}. It is shown in Appendix~\ref{app:trivariate} that $\Delta_3(\const{A},\alpha,\beta,\gamma)$ satisfies
\begin{equation*}
|\Delta_3(\const{A},\alpha,\beta,\gamma)| \leq \const{A}^2 \eta_3(\const{A},\alpha,\beta,\gamma)
\end{equation*}
where $\eta_3(\const{A},\alpha,\beta,\gamma)=\eta_3(\const{A},|\alpha|,|\beta|,|\gamma|)$ is monotonically increasing in $(\const{A},|\alpha|,|\beta|,|\gamma|)$ and is bounded for every finite $\const{A}$, $\alpha$, $\beta$, and $\gamma$.

Combining \eqref{eq:ternary_1}--\eqref{eq:ternary_5} yields
\begin{IEEEeqnarray}{lCl}
\IEEEeqnarraymulticol{3}{l}{\int_{-\alpha\const{A}}^{\infty}\int_{-\beta\const{A}}^{\infty}\int_{-\gamma\const{A}}^{\infty}\phi_{\vect{0},\mat{K}}(x,y,z)\d z\d y \d x}\nonumber\\
\quad & = & \frac{1}{8} + \frac{1}{4\pi}\bigl(\arcsin(\varrho_{12})+\arcsin(\varrho_{13})+\arcsin(\varrho_{23})\bigr)\nonumber\\
& & {} + \frac{\const{A}}{\sqrt{2\pi}}\left[\frac{\alpha+\beta+\gamma}{4} + \frac{\alpha}{2\pi}\arcsin\left(\frac{\varrho_{23}-\varrho_{12}\varrho_{13}}{\sqrt{(1-\varrho_{12}^2)(1-\varrho_{13}^2)}}\right)\right.\nonumber\\
& & \qquad\qquad\quad {} + \left.\frac{\beta}{2\pi}\arcsin\left(\frac{\varrho_{13}-\varrho_{12}\varrho_{23}}{\sqrt{(1-\varrho_{12}^2)(1-\varrho_{23}^2)}}\right)+\frac{\gamma}{2\pi}\arcsin\left(\frac{\varrho_{12}-\varrho_{13}\varrho_{23}}{\sqrt{(1-\varrho_{13}^2)(1-\varrho_{23}^2)}}\right)\right]\nonumber\\
& & {} + \Delta(\const{A},\alpha,\beta,\gamma) \label{eq:ternary_6}
\end{IEEEeqnarray}
where
\begin{equation*}
\Delta(\const{A},\alpha,\beta,\gamma) \triangleq \Delta_1(\const{A},\alpha,\beta,\gamma)+\Delta_2(\const{A},\alpha,\beta,\gamma) + \Delta_3(\const{A},\alpha,\beta,\gamma).
\end{equation*}
By the Triangle Inequality, we have
\begin{IEEEeqnarray}{lCl}
|\Delta(\const{A},\alpha,\beta,\gamma)| & \leq & |\Delta_1(\const{A},\alpha,\beta,\gamma)| + |\Delta_2(\const{A},\alpha,\beta,\gamma)| + |\Delta_3(\const{A},\alpha,\beta,\gamma)|\nonumber\\
& \leq & \const{A}^2 \bigl(\eta_1(\const{A},\alpha,\beta,\gamma) + \eta_2(\const{A},\alpha,\beta,\gamma) + \eta_3(\const{A},\alpha,\beta,\gamma)\bigr). \label{eq:ternary_7}
\end{IEEEeqnarray}
Proposition~\ref{prop:ternary} follows now from \eqref{eq:ternary_6} and \eqref{eq:ternary_7} by noting that if $\eta_1(\const{A},\alpha,\beta,\gamma)$, $\eta_2(\const{A},\alpha,\beta,\gamma)$, and $\eta_3(\const{A},\alpha,\beta,\gamma)$ are monotonically increasing in $(\const{A},|\alpha|,|\beta|,|\gamma|)$, then so is
\begin{equation*}
\eta(\const{A},\alpha,\beta,\gamma) \triangleq \eta_1(\const{A},\alpha,\beta,\gamma) + \eta_2(\const{A},\alpha,\beta,\gamma) + \eta_3(\const{A},\alpha,\beta,\gamma);
\end{equation*}
and if  $\eta_1(\const{A},\alpha,\beta,\gamma)$, $\eta_2(\const{A},\alpha,\beta,\gamma)$, and $\eta_3(\const{A},\alpha,\beta,\gamma)$ are bounded, then so is $\eta(\const{A},\alpha,\beta,\gamma)$.

\subsection{Proof of Theorem~\ref{thm:main}}
\label{sub:mainproof}
To prove Theorem~\ref{thm:main}, we derive a lower bound on $C_{\frac{1}{4\WW}}(\const{P})$ and compute its ratio to $\const{P}$ in the limit as $\const{P}$ tends to zero. To this end, we evaluate $(2\WW)/n\, I\bigl(X_1^n;\vect{Y}_1^n\bigr)$ for $\{X_k,\,k\in\Integers\}$ being a sequence of IID, binary random variables with
\begin{equation*}
X_k = \left\{\begin{array}{rl} \sqrt{P}, \qquad & \textnormal{with probability $\frac{1}{2}$} \\[5pt] -\sqrt{P}, \qquad & \textnormal{with probability $\frac{1}{2}$.} \end{array}\right.
\end{equation*}
We shall restrict ourselves to waveforms $g(\cdot)$ that satisfy
\begin{IEEEeqnarray}{rCl}
\sum_{\ell\neq 0} \left|g\left(\frac{\ell-1/2}{2\WW}\right)\right| & < & \infty \label{eq:g1}\\
\sum_{\ell\neq 0} \left|g\left(\frac{\ell}{2\WW}\right)\right| & < & \infty \label{eq:g2}\\
\sum_{\ell\neq 0} \left|g\left(\frac{\ell+1/2}{2\WW}\right)\right| & < & \infty. \label{eq:g3}
\end{IEEEeqnarray}
By the chain rule for mutual information \cite[Thm.~2.5.2]{coverthomas91}
\begin{IEEEeqnarray}{lCl}
\frac{2\WW}{n} I\bigl(X_1^n;\vect{Y}_1^n\bigr) & = & \frac{2\WW}{n} \sum_{k=1}^n I\bigl(X_k;\vect{Y}_1^n\bigm| X_1^{k-1}\bigr) \nonumber\\
& \geq & \frac{2\WW}{n} I\bigl(X_k;\vect{Y}_k\bigr)\nonumber\\
& = & \frac{2\WW}{n} \sum_{k=1}^n I\bigl(X_k;Y_{k-\frac{1}{2}},Y_k,Y_{k+\frac{1}{2}}\bigr)\nonumber\\
& = &  2\WW\, I\bigl(X_1;Y_{\frac{1}{2}},Y_1,Y_{\frac{3}{2}}\bigr) \label{eq:main_firststep}
\end{IEEEeqnarray}
where we define
\begin{equation*}
Y_{\tau} \triangleq Y\left(\frac{\tau}{2\WW}\right), \qquad \tau\in\Reals.
\end{equation*}
Here the second step follows because reducing observations cannot increase mutual information and because $\{X_k,\,k\in\Integers\}$ is IID; the third step follows from the definition of $\vect{Y}_k$; and the last step follows because the joint law of $\bigl(X_k,Y_{k-\frac{1}{2}},Y_k,Y_{k+\frac{1}{2}}\bigr)$ does not depend on $k$.

\subsubsection{The Joint Law of $\bigl(X_1,Y_{\frac{1}{2}},Y_1,Y_{\frac{3}{2}}\bigr)$}
In order to evaluate $I\bigl(X_1;Y_{\frac{1}{2}},Y_{1},Y_{\frac{3}{2}}\bigr)$, we shall compute the conditional probability of $\bigl(Y_{\frac{1}{2}},Y_1,Y_{\frac{3}{2}}\bigr)$, conditioned on $X_1$. To this end, we first compute the conditional probability of $\bigl(Y_{\frac{1}{2}},Y_1,Y_{\frac{3}{2}}\bigr)$, conditioned on $X_{-\infty}^{\infty}$, and then average over $(X_{-\infty}^{0},X_2^{\infty})$. 

To compute $\Prob\bigl(Y_{\frac{1}{2}}=1,Y_1=1,Y_{\frac{3}{2}}=1\bigm| X_{-\infty}^{\infty}=x_{-\infty}^{\infty}\bigr)$, we first note that by \eqref{eq:channel}
\begin{equation}
\label{eq:main_1}
Y_{\tau}=1 \qquad \Longleftrightarrow \qquad \frac{1}{\sqrt{2\WW}}\sum_{\ell=-\infty}^{\infty} x_{\ell} \, g\left(\frac{\tau-\ell}{2\WW}\right) + Z_{\tau} \geq 0, \qquad \tau\in\Reals
\end{equation}
where we define
\begin{equation*}
Z_{\tau} \triangleq \bigl(\conv{\vect{Z}}{\textnormal{LPF}_{\WW}}\bigr)\left(\frac{\tau}{2\WW}\right), \qquad \tau\in\Reals.
\end{equation*}
Let
\begin{IEEEeqnarray}{lCl}
\alpha\bigl(x_{-\infty}^{\infty}\bigr) & \triangleq & \frac{1}{\sqrt{\const{P}(2\WW)(\WW\Nzero)}} \sum_{\ell=-\infty}^{\infty} x_{\ell} \, g\left(\frac{1/2-\ell}{2\WW}\right)\nonumber\\
\beta\bigl(x_{-\infty}^{\infty}\bigr) & \triangleq & \frac{1}{\sqrt{\const{P}(2\WW)(\WW\Nzero)}} \sum_{\ell=-\infty}^{\infty} x_{\ell} \, g\left(\frac{1-\ell}{2\WW}\right)\nonumber
\end{IEEEeqnarray}
and
\begin{IEEEeqnarray}{lCl}
\gamma\bigl(x_{-\infty}^{\infty}\bigr) & \triangleq & \frac{1}{\sqrt{\const{P}(2\WW)(\WW\Nzero)}} \sum_{\ell=-\infty}^{\infty} x_{\ell} \, g\left(\frac{3/2-\ell}{2\WW}\right). \nonumber
\end{IEEEeqnarray}
It follows from \eqref{eq:main_1} that
\begin{IEEEeqnarray}{lCl}
\IEEEeqnarraymulticol{3}{l}{\Prob\bigl(Y_{\frac{1}{2}}=1,Y_1=1,Y_{\frac{3}{2}}=1\bigm| X_{-\infty}^{\infty}=x_{-\infty}^{\infty}\bigr)}\nonumber\\
\quad & = & \Prob\Bigl(Z_{\frac{1}{2}}\geq -\alpha\bigl(x_{-\infty}^{\infty}\bigr)\sqrt{\const{P}(\WW\Nzero)}, Z_1\geq -\beta\bigl(x_{-\infty}^{\infty}\bigr)\sqrt{\const{P}(\WW\Nzero)}, Z_{\frac{3}{2}} \geq -\gamma\bigl(x_{-\infty}^{\infty}\bigr)\sqrt{\const{P}(\WW\Nzero)}\Bigr)\nonumber\\
& = & \Prob\left(\frac{1}{\sqrt{\WW\Nzero}}Z_{\frac{1}{2}} \geq -\alpha\bigl(x_{-\infty}^{\infty}\bigr)\sqrt{\const{P}}, \frac{1}{\sqrt{\WW\Nzero}}Z_1\geq -\beta\bigl(x_{-\infty}^{\infty}\bigr)\sqrt{\const{P}}, \frac{1}{\sqrt{\WW\Nzero}}Z_{\frac{3}{2}} \geq -\gamma\bigl(x_{-\infty}^{\infty}\bigr)\sqrt{\const{P}}\right)\nonumber\\
& = & \int_{-\alpha(x_{-\infty}^{\infty})\sqrt{\const{P}}}^{\infty} \int_{-\beta(x_{-\infty}^{\infty})\sqrt{\const{P}}}^{\infty} \int_{-\gamma(x_{-\infty}^{\infty})\sqrt{\const{P}}}^{\infty} \phi_{\vect{0},\mat{K}} (x,y,z) \d z\d y\d x \nonumber\\
&= & \frac{1}{8} + \frac{1}{2\pi} \arcsin(\rho) + \sqrt{\frac{\const{P}}{2\pi}}\Biggl[\frac{\alpha\bigl(x_{-\infty}^{\infty}\bigr)+\beta\bigl(x_{-\infty}^{\infty}\bigr)+\gamma\bigl(x_{-\infty}^{\infty}\bigr)}{4}\nonumber\\
& & \qquad\qquad\qquad\qquad {} +\frac{\alpha\bigl(x_{-\infty}^{\infty}\bigr)+\gamma\bigl(x_{-\infty}^{\infty}\bigr)}{2\pi}\arcsin\left(\frac{\rho}{\sqrt{1-\rho^2}}\right)-\frac{\beta\bigl(x_{-\infty}^{\infty}\bigr)}{2\pi}\arcsin\left(\frac{\rho^2}{1-\rho^2}\right)\Biggr] \nonumber\\
& & {} + \Delta\left(\sqrt{\const{P}},\alpha\bigl(x_{-\infty}^{\infty}\bigr),\beta\bigl(x_{-\infty}^{\infty}\bigr),\gamma\bigl(x_{-\infty}^{\infty}\bigr)\right) \label{eq:main_2}
\end{IEEEeqnarray}
where
\begin{equation*}
\mat{K} = \left(\begin{array}{ccc} 1 & \rho & 0 \\ \rho & 1 & \rho \\ 0 & \rho & 1\end{array}\right), \qquad \rho=\frac{2}{\pi}.
\end{equation*}
Here the third step follows by noting that $\{Z_{\tau},\,\tau\in\Reals\}$ is a zero-mean Gaussian process of autocovariance function $\tau\mapsto\WW\Nzero\sinc(\tau)$; and the last step follows from Proposition~\ref{prop:ternary}. 

Since $\{X_k,\,k\in\Integers\}$ is IID and of zero mean, we have
\begin{IEEEeqnarray}{lCl}
\alpha_0 & \triangleq & \Econdd{\alpha\bigl(X_{-\infty}^{\infty}\bigr)}{X_1=\sqrt{\const{P}}}\nonumber\\
 & = & \frac{1}{\sqrt{(2\WW)(\WW\Nzero)}} g\left(-\frac{1}{4\WW}\right)+\frac{1}{\sqrt{\const{P}(2\WW)(\WW\Nzero)}}\sum_{\ell\neq 1} \E{X_{\ell}} g\left(\frac{1/2-\ell}{2\WW}\right)\nonumber\\
 & = & \frac{1}{\sqrt{(2\WW)(\WW\Nzero)}} g\left(-\frac{1}{4\WW}\right) \label{eq:main_alpha0}
 \end{IEEEeqnarray}
 and
 \begin{IEEEeqnarray}{lClCl}
\beta_0 & \triangleq & \Econdd{\beta\bigl(X_{-\infty}^{\infty}\bigr)}{X_1=\sqrt{\const{P}}} & = & \frac{1}{\sqrt{(2\WW)(\WW\Nzero)}} g\left(0\right)\label{eq:main_beta0}\\
\gamma_0 & \triangleq & \Econdd{\gamma\bigl(X_{-\infty}^{\infty}\bigr)}{X_1=\sqrt{\const{P}}} & = & 	\frac{1}{\sqrt{(2\WW)(\WW\Nzero)}} g\left(\frac{1}{4\WW}\right).\label{eq:main_gamma0}
\end{IEEEeqnarray}
By setting $X_1=\sqrt{\const{P}}$ and averaging over $(X_{-\infty}^0,X_2^{\infty})$, we thus obtain
\begin{IEEEeqnarray}{lCl}
\IEEEeqnarraymulticol{3}{l}{\Prob\bigl(Y_{\frac{1}{2}}=1,Y_1=1,Y_{\frac{3}{2}}=1\bigm| X_1=\sqrt{\const{P}}\bigr)}\nonumber\\
\quad & = & \frac{1}{8} + \frac{1}{2\pi} \arcsin(\rho) + \Expec\Biggl[\sqrt{\frac{\const{P}}{2\pi}}\Biggl[\frac{\alpha\bigl(X_{-\infty}^{\infty}\bigr)+\beta\bigl(X_{-\infty}^{\infty}\bigr)+\gamma\bigl(X_{-\infty}^{\infty}\bigr)}{4} \nonumber\\
& & \quad {}  +\frac{\alpha\bigl(X_{-\infty}^{\infty}\bigr)+\gamma\bigl(X_{-\infty}^{\infty}\bigr)}{2\pi}\arcsin\left(\frac{\rho}{\sqrt{1-\rho^2}}\right)-\frac{\beta\bigl(X_{-\infty}^{\infty}\bigr)}{2\pi}\arcsin\left(\frac{\rho^2}{1-\rho^2}\right)\Biggr] \Biggm| X_1 = \sqrt{\const{P}} \Biggr]\nonumber\\
& & {} + \Econd{\Delta\left(\sqrt{\const{P}},\alpha\bigl(X_{-\infty}^{\infty}\bigr),\beta\bigl(X_{-\infty}^{\infty}\bigr),\gamma\bigl(X_{-\infty}^{\infty}\bigr)\right)}{X_1=\sqrt{\const{P}}}\nonumber\\
& = & \frac{1}{8} + \frac{1}{2\pi} \arcsin(\rho) + \sqrt{\frac{\const{P}}{2\pi}}\Biggl[\frac{\alpha_0+\beta_0+\gamma_0}{4}\nonumber\\
& & \qquad\qquad\qquad\qquad\qquad\qquad {} +\frac{\alpha_0+\gamma_0}{2\pi}\arcsin\left(\frac{\rho}{\sqrt{1-\rho^2}}\right)-\frac{\beta_0}{2\pi}\arcsin\left(\frac{\rho^2}{1-\rho^2}\right)\Biggr] \nonumber\\
& & {} + \Econd{\Delta\left(\sqrt{\const{P}},\alpha\bigl(X_{-\infty}^{\infty}\bigr),\beta\bigl(X_{-\infty}^{\infty}\bigr),\gamma\bigl(X_{-\infty}^{\infty}\bigr)\right)}{X_1=\sqrt{\const{P}}}.\label{eq:main_3}
\end{IEEEeqnarray}
We next show that
\begin{equation}
\Econd{\Delta\left(\sqrt{\const{P}},\alpha\bigl(X_{-\infty}^{\infty}\bigr),\beta\bigl(X_{-\infty}^{\infty}\bigr),\gamma\bigl(X_{-\infty}^{\infty}\bigr)\right)}{X_1=\sqrt{\const{P}}} = o\bigl(\sqrt{\const{P}}\bigr) \label{eq:main_Delta}
\end{equation}
where $x\mapsto o(x)$ satisfies $\lim_{x\downarrow0}o(x)/x=0$. To this end, we use the Triangle Inequality to upper bound
\begin{IEEEeqnarray}{lCl}
\left|\alpha\bigl(x_{-\infty}^{\infty}\bigr)\right| & \leq & \frac{1}{\sqrt{\const{P}(2\WW)(\WW\Nzero)}}\sum_{\ell=-\infty}^{\infty} |x_{\ell}| \left|g\left(\frac{1/2-\ell}{2\WW}\right)\right| \nonumber\\
& = & \frac{1}{\sqrt{(2\WW)(\WW\Nzero)}} \sum_{\ell=-\infty}^{\infty} \left|g\left(\frac{1/2-\ell}{2\WW}\right)\right| \nonumber\\
& \triangleq & \alpha_{\textnormal{max}}\nonumber
\end{IEEEeqnarray}
which, by \eqref{eq:g1}, is finite. Here the second step follows because, for our choice of $\{X_{k},\,k\in\Integers\}$, we have with probability one $|X_k|=\sqrt{\const{P}}$. Similarly, we have
\begin{IEEEeqnarray}{lCl}
\left|\beta\bigl(x_{-\infty}^{\infty}\bigr)\right| & \leq & \frac{1}{\sqrt{(2\WW)(\WW\Nzero)}} \sum_{\ell=-\infty}^{\infty} \left|g\left(\frac{1-\ell}{2\WW}\right)\right| \nonumber\\
& \triangleq & \beta_{\textnormal{max}} \nonumber\\
& < & \infty\nonumber
\end{IEEEeqnarray}
and
\begin{IEEEeqnarray}{lCl}
\left|\gamma\bigl(x_{-\infty}^{\infty}\bigr)\right| & \leq & \frac{1}{\sqrt{(2\WW)(\WW\Nzero)}} \sum_{\ell=-\infty}^{\infty} \left|g\left(\frac{3/2-\ell}{2\WW}\right)\right| \nonumber\\
& \triangleq & \gamma_{\textnormal{max}}\nonumber\\
& < & \infty.\nonumber
\end{IEEEeqnarray}
We thus obtain
\begin{IEEEeqnarray}{lCl}
\IEEEeqnarraymulticol{3}{l}{\left|\Econd{\Delta\left(\sqrt{\const{P}},\alpha\bigl(X_{-\infty}^{\infty}\bigr),\beta\bigl(X_{-\infty}^{\infty}\bigr),\gamma\bigl(X_{-\infty}^{\infty}\bigr)\right)}{X_1=\sqrt{\const{P}}}\right|}\nonumber\\
\quad & \leq & \Econd{\left|\Delta\left(\sqrt{\const{P}},\alpha\bigl(X_{-\infty}^{\infty}\bigr),\beta\bigl(X_{-\infty}^{\infty}\bigr),\gamma\bigl(X_{-\infty}^{\infty}\bigr)\right)\right|}{X_1=\sqrt{\const{P}}}\nonumber\\
& \leq & \const{P} \,\Econd{\eta\left(\sqrt{\const{P}},\alpha\bigl(X_{-\infty}^{\infty}\bigr),\beta\bigl(X_{-\infty}^{\infty}\bigr),\gamma\bigl(X_{-\infty}^{\infty}\bigr)\right)}{X_1=\sqrt{\const{P}}} \nonumber\\
& \leq & \const{P}\, \eta\left(\sqrt{\const{P}},\alpha_{\textnormal{max}},\beta_{\textnormal{max}},\gamma_{\textnormal{max}}\right) \nonumber\\
& = & o\bigl(\sqrt{\const{P}}\bigr) \label{eq:main_step2Delta}
\end{IEEEeqnarray}
where the first step follows from the Triangle Inequality; the second step follows because $\Delta(\const{A},\alpha,\beta,\gamma)\leq\const{A}^2\eta(\const{A},\alpha,\beta,\gamma)$; the third step follows because $\eta(\const{A},\alpha,\beta,\gamma)=\eta(\const{A},|\alpha|,|\beta|,|\gamma|)$ is monotonically increasing in $(\const{A},|\alpha|,|\beta|,|\gamma|)$; and the last step follows because $\eta(\const{A},\alpha,\beta,\gamma)$ is bounded for finite $\const{A}$, $\alpha$, $\beta$, and $\gamma$. This proves \eqref{eq:main_Delta}.

Combining \eqref{eq:main_3} and \eqref{eq:main_Delta} yields
\begin{IEEEeqnarray}{lCl}
\IEEEeqnarraymulticol{3}{l}{\Prob\bigl(Y_{\frac{1}{2}}=1,Y_1=1,Y_{\frac{3}{2}}=1\bigm| X_1=\sqrt{\const{P}}\bigr)}\nonumber\\
\, & = & \frac{1}{8} + \frac{1}{2\pi} \arcsin(\rho)\nonumber\\
& & {}  + \sqrt{\frac{\const{P}}{2\pi}}\Biggl[\frac{\alpha_0+\beta_0+\gamma_0}{4} +\frac{\alpha_0+\gamma_0}{2\pi}\arcsin\left(\frac{\rho}{\sqrt{1-\rho^2}}\right)-\frac{\beta_0}{2\pi}\arcsin\left(\frac{\rho^2}{1-\rho^2}\right)\Biggr] + o\bigl(\sqrt{\const{P}}\bigr).\IEEEeqnarraynumspace\label{eq:main_111|A}
\end{IEEEeqnarray}
By averaging \eqref{eq:main_2} over $X_{-\infty}^{\infty}$, and by noting that
\begin{equation}
\E{\alpha\bigl(X_{-\infty}^{\infty}\bigr)} = \E{\beta\bigl(X_{-\infty}^{\infty}\bigr)} = \E{\gamma\bigl(X_{-\infty}^{\infty}\bigr)} = 0 \label{eq:main_unconditioned}
\end{equation}
and
\begin{equation*}
\E{\Delta\left(\sqrt{\const{P}},\alpha\bigl(X_{-\infty}^{\infty}\bigr),\beta\bigl(X_{-\infty}^{\infty}\bigr),\gamma\bigl(X_{-\infty}^{\infty}\bigr)\right)} = o\bigl(\sqrt{\const{P}}\bigr)
\end{equation*}
we obtain  for the unconditional probability
\begin{IEEEeqnarray}{lCl}
\Prob\bigl(Y_{\frac{1}{2}}=1,Y_1=1,Y_{\frac{3}{2}}=1\bigr) & = & \frac{1}{8} + \frac{1}{2\pi} \arcsin(\rho) + o\bigl(\sqrt{\const{P}}\bigr). \label{eq:main_111}
\end{IEEEeqnarray}
It thus follows from Bayes' law that
\begin{IEEEeqnarray}{lCl}
\IEEEeqnarraymulticol{3}{l}{\Prob\bigl(X_1=\sqrt{\const{P}} \bigm| Y_{\frac{1}{2}}=1,Y_1=1,Y_{\frac{3}{2}}=1\bigr)}\nonumber\\
\quad & = & \frac{1}{2} + \sqrt{\frac{\const{P}}{2\pi}}\frac{\frac{\alpha_0+\beta_0+\gamma_0}{4} +\frac{\alpha_0+\gamma_0}{2\pi}\arcsin\left(\frac{\rho}{\sqrt{1-\rho^2}}\right)-\frac{\beta_0}{2\pi}\arcsin\left(\frac{\rho^2}{1-\rho^2}\right)+o\bigl(\sqrt{\const{P}}\bigr)}{\frac{1}{4} + \frac{1}{\pi} \arcsin(\rho) + o\bigl(\sqrt{\const{P}}\bigr)} \nonumber\\
& = & \frac{1}{2} + \sqrt{\frac{\const{P}}{2\pi}}\frac{\frac{\alpha_0+\beta_0+\gamma_0}{4} +\frac{\alpha_0+\gamma_0}{2\pi}\arcsin\left(\frac{\rho}{\sqrt{1-\rho^2}}\right)-\frac{\beta_0}{2\pi}\arcsin\left(\frac{\rho^2}{1-\rho^2}\right)}{\frac{1}{4} + \frac{1}{\pi} \arcsin(\rho)} + o\bigl(\sqrt{\const{P}}\bigr)\label{eq:main_A|111}
\end{IEEEeqnarray}
where the last step follows because $\frac{\vartheta+o(x)}{\xi+o(x)}=\frac{\vartheta}{\xi}+o(x)$ for any $\xi\neq 0$ and $\vartheta\in\Reals$. Note that the first two terms on the RHS of \eqref{eq:main_A|111} depend on $\ldots,X_{-1},X_{0},X_2,X_3,\ldots$ only via their means. Thus, if $\{X_k,\,k\in\Integers\}$ is of zero mean, then intersymbol interference affects only the $o\bigl(\sqrt{\const{P}}\bigr)$-term, which does not influence our lower bound on the capacity per unit-cost.

The probability $\Prob\bigl(Y_{\frac{1}{2}}=1,Y_1=1,Y_{\frac{3}{2}}=-1\bigm| X_1=\sqrt{\const{P}}\bigr)$ can be computed in a similar way. We have
\begin{IEEEeqnarray}{lCl}
\IEEEeqnarraymulticol{3}{l}{\Prob\bigl(Y_{\frac{1}{2}}=1,Y_1=1,Y_{\frac{3}{2}}=-1\bigm| X_{-\infty}^{\infty}=x_{-\infty}^{\infty}\bigr)}\nonumber\\
\quad & = & \int_{-\alpha(x_{-\infty}^{\infty})\sqrt{\const{P}}}^{\infty} \int_{-\beta(x_{-\infty}^{\infty})\sqrt{\const{P}}}^{\infty} \int_{-\infty}^{-\gamma(x_{-\infty}^{\infty})\sqrt{\const{P}}} \phi_{\vect{0},\mat{K}} (x,y,z) \d z\d y\d x \nonumber\\
& = & \int_{-\alpha(x_{-\infty}^{\infty})\sqrt{\const{P}}}^{\infty} \int_{-\beta(x_{-\infty}^{\infty})\sqrt{\const{P}}}^{\infty} \phi_{\vect{0},\mat{G}} (x,y) \d y\d x \nonumber\\
& & {} - \int_{-\alpha(x_{-\infty}^{\infty})\sqrt{\const{P}}}^{\infty} \int_{-\beta(x_{-\infty}^{\infty})\sqrt{\const{P}}}^{\infty} \int_{-\gamma(x_{-\infty}^{\infty})\sqrt{\const{P}}}^{\infty} \phi_{\vect{0},\mat{K}} (x,y,z) \d z\d y\d x \nonumber\\
& = & \frac{1}{4} + \frac{1}{2\pi}\arcsin(\rho) + \frac{\alpha(x_{-\infty}^{\infty})+\beta(x_{-\infty}^{\infty})}{2}\sqrt{\frac{\const{P}}{2\pi}} + \Delta\left(\sqrt{\const{P}},\alpha(x_{-\infty}^{\infty}),\beta(x_{-\infty}^{\infty})\right) \nonumber\\
& & {} - \frac{1}{8} - \frac{1}{2\pi} \arcsin(\rho) - \sqrt{\frac{\const{P}}{2\pi}}\Biggl[\frac{\alpha\bigl(x_{-\infty}^{\infty}\bigr)+\beta\bigl(x_{-\infty}^{\infty}\bigr)+\gamma\bigl(x_{-\infty}^{\infty}\bigr)}{4}\nonumber\\
& & \qquad\qquad\qquad\qquad {} +\frac{\alpha\bigl(x_{-\infty}^{\infty}\bigr)+\gamma\bigl(x_{-\infty}^{\infty}\bigr)}{2\pi}\arcsin\left(\frac{\rho}{\sqrt{1-\rho^2}}\right)-\frac{\beta\bigl(x_{-\infty}^{\infty}\bigr)}{2\pi}\arcsin\left(\frac{\rho^2}{1-\rho^2}\right)\Biggr] \nonumber\\
& & {} - \Delta\left(\sqrt{\const{P}},\alpha\bigl(x_{-\infty}^{\infty}\bigr),\beta\bigl(x_{-\infty}^{\infty}\bigr),\gamma\bigl(x_{-\infty}^{\infty}\bigr)\right)\nonumber\\
& = & \frac{1}{8} + \sqrt{\frac{\const{P}}{2\pi}}\Biggl[\frac{\alpha\bigl(x_{-\infty}^{\infty}\bigr)+\beta\bigl(x_{-\infty}^{\infty}\bigr)-\gamma\bigl(x_{-\infty}^{\infty}\bigr)}{4}\nonumber\\
& & \qquad\qquad\quad {} -\frac{\alpha\bigl(x_{-\infty}^{\infty}\bigr)+\gamma\bigl(x_{-\infty}^{\infty}\bigr)}{2\pi}\arcsin\left(\frac{\rho}{\sqrt{1-\rho^2}}\right)+\frac{\beta\bigl(x_{-\infty}^{\infty}\bigr)}{2\pi}\arcsin\left(\frac{\rho^2}{1-\rho^2}\right)\Biggr] \nonumber\\
& & {} + \Delta\left(\sqrt{\const{P}},\alpha(x_{-\infty}^{\infty}),\beta(x_{-\infty}^{\infty})\right) - \Delta\left(\sqrt{\const{P}},\alpha\bigl(x_{-\infty}^{\infty}\bigr),\beta\bigl(x_{-\infty}^{\infty}\bigr),\gamma\bigl(x_{-\infty}^{\infty}\bigr)\right)
\end{IEEEeqnarray}
where $\mat{K}$ is the same as in \eqref{eq:main_2}, and where
\begin{equation*}
\mat{G} = \left(\begin{array}{cc} 1 & \rho \\ \rho & 1\end{array}\right), \qquad \rho=\frac{2}{\pi}.
\end{equation*}
Here the second step follows because
\begin{equation*}
\int_{-\infty}^{\infty} \phi_{0,\mat{K}}(x,y,z)\d z = \phi_{0,\mat{G}}(x,y), \qquad \bigl(x\in\Reals,\, y\in\Reals\bigr)
\end{equation*}
and the third step follows from Propositions~\ref{prop:binary} and \ref{prop:ternary}. By setting $X_1=\sqrt{\const{P}}$ and averaging over $\bigl(X_{-\infty}^0,X_2^{\infty}\bigr)$, we obtain
\begin{IEEEeqnarray}{lCl}
\IEEEeqnarraymulticol{3}{l}{\Prob\bigl(Y_{\frac{1}{2}}=1,Y_1=1,Y_{\frac{3}{2}}=-1\bigm| X_1=\sqrt{\const{P}}\bigr)}\nonumber\\
\quad & = & \frac{1}{8} + \sqrt{\frac{\const{P}}{2\pi}}\Biggl[\frac{\alpha_0+\beta_0-\gamma_0}{4}  -\frac{\alpha_0+\gamma_0}{2\pi}\arcsin\left(\frac{\rho}{\sqrt{1-\rho^2}}\right)+\frac{\beta_0}{2\pi}\arcsin\left(\frac{\rho^2}{1-\rho^2}\right)\Biggr] \nonumber\\
& & {} + \Econd{\Delta\left(\sqrt{\const{P}},\alpha(x_{-\infty}^{\infty}),\beta(x_{-\infty}^{\infty})\right)}{X_1=\sqrt{\const{P}}}\nonumber\\
& & {}  - \Econd{\Delta\left(\sqrt{\const{P}},\alpha\bigl(x_{-\infty}^{\infty}\bigr),\beta\bigl(x_{-\infty}^{\infty}\bigr),\gamma\bigl(x_{-\infty}^{\infty}\bigr)\right)}{\sqrt{\const{P}}} \label{eq:main_4}
\end{IEEEeqnarray}
where $\alpha_0$, $\beta_0$, and $\gamma_0$ are defined in \eqref{eq:main_alpha0}--\eqref{eq:main_gamma0}. It was shown above that the last term on the RHS of \eqref{eq:main_4} decays faster than $\sqrt{\const{P}}$ \eqref{eq:main_Delta}. Repeating the steps in \eqref{eq:main_step2Delta}, it can also be shown that
\begin{equation*}
\Econd{\Delta\left(\sqrt{\const{P}},\alpha(x_{-\infty}^{\infty}),\beta(x_{-\infty}^{\infty})\right)}{X_1=\sqrt{\const{P}}} = o\bigl(\sqrt{\const{P}}\bigr).
\end{equation*}
By noting that $o\bigl(\sqrt{\const{P}}\bigr)-o\bigl(\sqrt{\const{P}}\bigr)=o\bigl(\sqrt{\const{P}}\bigr)$, we thus obtain
\begin{IEEEeqnarray}{lCl}
\Prob\bigl(Y_{\frac{1}{2}}=1,Y_1=1,Y_{\frac{3}{2}}=-1\bigm| X_1=\sqrt{\const{P}}\bigr) & = & \frac{1}{8}\nonumber\\
\IEEEeqnarraymulticol{3}{r}{ \,\,\, {} + \sqrt{\frac{\const{P}}{2\pi}}\Biggl[\frac{\alpha_0+\beta_0-\gamma_0}{4}  -\frac{\alpha_0+\gamma_0}{2\pi}\arcsin\left(\frac{\rho}{\sqrt{1-\rho^2}}\right)+\frac{\beta_0}{2\pi}\arcsin\left(\frac{\rho^2}{1-\rho^2}\right)\Biggr]  + o\bigl(\sqrt{\const{P}}\bigr). \label{eq:main_11-1|A}\IEEEeqnarraynumspace}
\end{IEEEeqnarray}
It follows from \eqref{eq:main_unconditioned} that the unconditional probability is given by
\begin{IEEEeqnarray}{lCl}
\Prob\bigl(Y_{\frac{1}{2}}=1,Y_1=1,Y_{\frac{3}{2}}=-1\bigr) & = & \frac{1}{8} + o\bigl(\sqrt{\const{P}}\bigr)\label{eq:main_11-1}.
\end{IEEEeqnarray}
By Bayes' law, we thus have
\begin{IEEEeqnarray}{lCl}
\Prob\bigl(X_1=\sqrt{\const{P}} \bigm| Y_{\frac{1}{2}}=1,Y_1=1,Y_{\frac{3}{2}}=-1\bigr) & = & \frac{1}{2}\nonumber\\
\IEEEeqnarraymulticol{3}{r}{{} + 4\sqrt{\frac{\const{P}}{2\pi}}\Biggl[\frac{\alpha_0+\beta_0-\gamma_0}{4}  -\frac{\alpha_0+\gamma_0}{2\pi}\arcsin\left(\frac{\rho}{\sqrt{1-\rho^2}}\right)+\frac{\beta_0}{2\pi}\arcsin\left(\frac{\rho^2}{1-\rho^2}\right)\Biggr] + o\bigl(\sqrt{\const{P}}\bigr). \IEEEeqnarraynumspace} \label{eq:main_A|11-1}.
\end{IEEEeqnarray}

The other conditional probabilities can be computed along the same lines. The probabilities corresponding to $\bigl(Y_{\frac{1}{2}}=1,Y_1=-1,Y_{\frac{3}{2}}=1\bigr)$ and $\bigl(Y_{\frac{1}{2}}=-1,Y_1=1,Y_{\frac{3}{2}}=1\bigr)$ are given by
\begin{IEEEeqnarray}{lCl}
\IEEEeqnarraymulticol{3}{l}{\Prob\bigl(X_1=\sqrt{\const{P}} \bigm| Y_{\frac{1}{2}}=1,Y_1=-1,Y_{\frac{3}{2}}=1\bigr)}\nonumber\\
\quad & = & \frac{1}{2} + \sqrt{\frac{\const{P}}{2\pi}}\frac{\frac{\alpha_0-\beta_0+\gamma_0}{4}-\frac{\alpha_0+\gamma_0}{2\pi}\arcsin\left(\frac{\rho}{\sqrt{1-\rho^2}}\right)+\frac{\beta_0}{2\pi}\arcsin\left(\frac{\rho^2}{1-\rho^2}\right)}{\frac{1}{4}-\frac{1}{\pi}\arcsin(\rho)} + o\bigl(\sqrt{\const{P}}\bigr)\label{eq:main_A|1-11} 
\end{IEEEeqnarray}
and
\begin{IEEEeqnarray}{lCl}
\Prob\bigl(X_1=\sqrt{\const{P}} \bigm| Y_{\frac{1}{2}}=-1,Y_1=1,Y_{\frac{3}{2}}=1\bigr) & = & \frac{1}{2}\nonumber\\
\IEEEeqnarraymulticol{3}{r}{{} + 4\sqrt{\frac{\const{P}}{2\pi}}\Biggl[\frac{-\alpha_0+\beta_0+\gamma_0}{4}  -\frac{\alpha_0+\gamma_0}{2\pi}\arcsin\left(\frac{\rho}{\sqrt{1-\rho^2}}\right)+\frac{\beta_0}{2\pi}\arcsin\left(\frac{\rho^2}{1-\rho^2}\right)\Biggr] + o\bigl(\sqrt{\const{P}}\bigr). \IEEEeqnarraynumspace} \label{eq:main_A|-111}
\end{IEEEeqnarray}
The remaining probabilities can be computed from \eqref{eq:main_A|111}, \eqref{eq:main_A|11-1}, \eqref{eq:main_A|1-11}, and \eqref{eq:main_A|-111} by noting that, due to the symmetry of $\phi_{\vect{0},\mat{K}}(\cdot)$,
\begin{equation*}
\Prob\bigl(X_1=\sqrt{\const{P}} \bigm| Y_{\frac{1}{2}}=y_{\frac{1}{2}},Y_1=y_1,Y_{\frac{3}{2}}=y_{\frac{3}{2}}\bigr) = 
\Prob\bigl(X_1=-\sqrt{\const{P}} \bigm| Y_{\frac{1}{2}}=-y_{\frac{1}{2}},Y_1=-y_1,Y_{\frac{3}{2}}=-y_{\frac{3}{2}}\bigr).
\end{equation*}

\subsubsection{Evaluating $I\bigl(X_1;Y_{\frac{1}{2}},Y_{1},Y_{\frac{3}{2}}\bigr)$}
Let
\begin{equation*}
\wp\bigl(y_{\frac{1}{2}},y_1,y_{\frac{3}{2}}\bigr) \triangleq \Prob\bigl(X_1=\sqrt{\const{P}}\bigm| Y_{\frac{1}{2}}=y_{\frac{1}{2}}, Y_1=y_1, Y_{\frac{3}{2}}=y_{\frac{3}{2}}\bigr).
\end{equation*}
We express the mutual information $I\bigl(X_1;Y_{\frac{1}{2}},Y_{1},Y_{\frac{3}{2}}\bigr)$ as
\begin{IEEEeqnarray}{lCl}
I\bigl(X_1;Y_{\frac{1}{2}},Y_{1},Y_{\frac{3}{2}}\bigr) & = & H\bigl(X_1\bigr) - H\bigl(X_1\bigm| Y_{\frac{1}{2}},Y_{1},Y_{\frac{3}{2}}\bigr)\nonumber\\
& = & \log 2 - \E{H_b\Bigl(\wp\bigl(Y_{\frac{1}{2}},Y_1,Y_{\frac{3}{2}}\bigr)\Bigr)}
\end{IEEEeqnarray}
where $H_b(p)\triangleq-p\log p-(1-p)\log(1-p)$, $0\leq p\leq 1$ (with $0\log 0 \triangleq 0$) denotes the binary entropy function. To evaluate the binary entropy function, we express $H_b(\cdot)$ as a Taylor series expansion around $\frac{1}{2}$, i.e.,
\begin{equation}
H_b\left(\frac{1}{2}+\xi\right) = \log 2 - 2\xi^2+o\bigl(\xi^2\bigr), \qquad |\xi|\leq\frac{1}{2}. \label{eq:main_TaylorHb}
\end{equation}
Substituting $\xi=\wp\bigl(y_{\frac{1}{2}},y_1,y_{\frac{3}{2}}\bigr)-1/2$ and averaging over $\bigl(Y_{\frac{1}{2}},Y_1,Y_{\frac{3}{2}}\bigr)$ yields thus
\begin{IEEEeqnarray}{lCl}
I\bigl(X_1;Y_{\frac{1}{2}},Y_{1},Y_{\frac{3}{2}}\bigr) & = & \log 2 - \E{\log 2-2\left(\wp\bigl(Y_{\frac{1}{2}},Y_1,Y_{\frac{3}{2}}\bigr)-\frac{1}{2}\right)^2+o\left(\biggl(\wp\bigl(Y_{\frac{1}{2}},Y_1,Y_{\frac{3}{2}}\bigr)-\frac{1}{2}\biggr)^2\right)} \nonumber\\
& = & 2\,\E{\left(\wp\bigl(Y_{\frac{1}{2}},Y_1,Y_{\frac{3}{2}}\bigr)-\frac{1}{2}\right)^2} +\E{o\left(\biggl(\wp\bigl(Y_{\frac{1}{2}},Y_1,Y_{\frac{3}{2}}\bigr)-\frac{1}{2}\biggr)^2\right)} \nonumber\\
& = & 2\,\E{\left(\wp\bigl(Y_{\frac{1}{2}},Y_1,Y_{\frac{3}{2}}\bigr)-\frac{1}{2}\right)^2} + o(\const{P})\label{eq:main_beforeavg}
\end{IEEEeqnarray}
where the last step follows because, by \eqref{eq:main_A|111}, \eqref{eq:main_A|11-1}, \eqref{eq:main_A|1-11}, and \eqref{eq:main_A|-111}, we have for every $\bigl(y_{\frac{1}{2}},y_1,y_{\frac{3}{2}}\bigr)$
\begin{equation*}
\left(\wp\bigl(y_{\frac{1}{2}},y_1,y_{\frac{3}{2}}\bigr) - \frac{1}{2}\right)^2 = \mathcal{O}(\const{P})
\end{equation*}
where $x\mapsto\mathcal{O}(x)$ satisfies $\lim_{x\downarrow 0} \bigl|\mathcal{O}(x)/x\bigr| < \infty$, so
\begin{equation*}
o\left(\biggl(\wp\bigl(y_{\frac{1}{2}},y_1,y_{\frac{3}{2}}\bigr) - \frac{1}{2}\biggr)^2\right) = o(\const{P})
\end{equation*}
which implies that
\begin{equation*}
\E{o\left(\biggl(\wp\bigl(Y_{\frac{1}{2}},Y_1,Y_{\frac{3}{2}}\bigr) - \frac{1}{2}\biggr)^2\right)} = o(\const{P})
\end{equation*}
since the expectation is given by the sum of eight terms that all decay faster than $\const{P}$.

By applying \eqref{eq:main_A|111}, \eqref{eq:main_A|11-1}, \eqref{eq:main_A|1-11}, and \eqref{eq:main_A|-111} to \eqref{eq:main_beforeavg}, we obtain
\begin{IEEEeqnarray}{lCl}
\IEEEeqnarraymulticol{3}{l}{I\bigl(X_1;Y_{\frac{1}{2}},Y_{1},Y_{\frac{3}{2}}\bigr)}\nonumber\\
& = & \frac{\const{P}}{\pi} \left[\frac{\left(\frac{\alpha_0+\beta_0+\gamma_0}{4} +\frac{\alpha_0+\gamma_0}{2\pi}\arcsin\left(\frac{\rho}{\sqrt{1-\rho^2}}\right)-\frac{\beta_0}{2\pi}\arcsin\left(\frac{\rho^2}{1-\rho^2}\right)\right)^2}{\frac{1}{4} + \frac{1}{\pi} \arcsin(\rho)} \right.\nonumber\\
& & \quad {} + 4\,\Biggl(\frac{\alpha_0+\beta_0-\gamma_0}{4}  -\frac{\alpha_0+\gamma_0}{2\pi}\arcsin\left(\frac{\rho}{\sqrt{1-\rho^2}}\right)+\frac{\beta_0}{2\pi}\arcsin\left(\frac{\rho^2}{1-\rho^2}\right)\Biggr)^2 \nonumber\\
& & \quad {} +  \frac{\left(\frac{\alpha_0-\beta_0+\gamma_0}{4}-\frac{\alpha_0+\gamma_0}{2\pi}\arcsin\left(\frac{\rho}{\sqrt{1-\rho^2}}\right)+\frac{\beta_0}{2\pi}\arcsin\left(\frac{\rho^2}{1-\rho^2}\right)\right)^2}{\frac{1}{4}-\frac{1}{\pi}\arcsin(\rho)}\nonumber\\
& & \quad {} + \left.4\,\Biggl(\frac{-\alpha_0+\beta_0+\gamma_0}{4}  -\frac{\alpha_0+\gamma_0}{2\pi}\arcsin\left(\frac{\rho}{\sqrt{1-\rho^2}}\right)+\frac{\beta_0}{2\pi}\arcsin\left(\frac{\rho^2}{1-\rho^2}\right)\Biggr)^2 \vphantom{\frac{\left(\frac{\alpha_0-\beta_0+\gamma_0}{4}-\frac{\alpha_0+\gamma_0}{2\pi}\arcsin\left(\frac{\rho}{\sqrt{1-\rho^2}}\right)+\frac{\beta_0}{2\pi}\arcsin\left(\frac{\rho^2}{1-\rho^2}\right)\right)^2}{\frac{1}{4}-\frac{1}{\pi}\arcsin(\rho)}}\right]+o(\const{P})\IEEEeqnarraynumspace\label{eq:main_LBMI}
\end{IEEEeqnarray}
where $\alpha_0$, $\beta_0$, and $\gamma_0$ are defined in \eqref{eq:main_alpha0}--\eqref{eq:main_gamma0}. 

Combining \eqref{eq:main_LBMI} with \eqref{eq:main_firststep} and \eqref{eq:capacity}, and computing the ratio to $\const{P}$ in the limit as $\const{P}$ tends to zero, yields the lower bound on the capacity per unit-cost
\begin{IEEEeqnarray}{lCl}
\dot{C}_{\frac{1}{4\WW}}(0)
 & \geq & \frac{2\WW}{\pi} \left[\frac{\left(\frac{\alpha_0+\beta_0+\gamma_0}{4} +\frac{\alpha_0+\gamma_0}{2\pi}\arcsin\left(\frac{\rho}{\sqrt{1-\rho^2}}\right)-\frac{\beta_0}{2\pi}\arcsin\left(\frac{\rho^2}{1-\rho^2}\right)\right)^2}{\frac{1}{4} + \frac{1}{\pi} \arcsin(\rho)} \right.\nonumber\\
& &  \qquad {} + 4\,\Biggl(\frac{\alpha_0+\beta_0-\gamma_0}{4}  -\frac{\alpha_0+\gamma_0}{2\pi}\arcsin\left(\frac{\rho}{\sqrt{1-\rho^2}}\right)+\frac{\beta_0}{2\pi}\arcsin\left(\frac{\rho^2}{1-\rho^2}\right)\Biggr)^2 \nonumber\\
& &  \qquad {} +  \frac{\left(\frac{\alpha_0-\beta_0+\gamma_0}{4}-\frac{\alpha_0+\gamma_0}{2\pi}\arcsin\left(\frac{\rho}{\sqrt{1-\rho^2}}\right)+\frac{\beta_0}{2\pi}\arcsin\left(\frac{\rho^2}{1-\rho^2}\right)\right)^2}{\frac{1}{4}-\frac{1}{\pi}\arcsin(\rho)}\nonumber\\
& &  \qquad {} + \left.4\,\Biggl(\frac{-\alpha_0+\beta_0+\gamma_0}{4}  -\frac{\alpha_0+\gamma_0}{2\pi}\arcsin\left(\frac{\rho}{\sqrt{1-\rho^2}}\right)+\frac{\beta_0}{2\pi}\arcsin\left(\frac{\rho^2}{1-\rho^2}\right)\Biggr)^2 \vphantom{\frac{\left(\frac{\alpha_0-\beta_0+\gamma_0}{4}-\frac{\alpha_0+\gamma_0}{2\pi}\arcsin\left(\frac{\rho}{\sqrt{1-\rho^2}}\right)+\frac{\beta_0}{2\pi}\arcsin\left(\frac{\rho^2}{1-\rho^2}\right)\right)^2}{\frac{1}{4}-\frac{1}{\pi}\arcsin(\rho)}}\right]\nonumber\\
& \triangleq & \dot{R}(0).\label{eq:main_LBgeneral}
\end{IEEEeqnarray}
Note that this lower bound holds for all unit-energy waveforms $g(\cdot)$ that are bandlimited to $\WW$ Hz and that satisfy \eqref{eq:g1}--\eqref{eq:g3}. In the following section we evaluate the RHS of \eqref{eq:main_LBgeneral} for a specific choice of $g(\cdot)$.

\subsubsection{Choosing a Waveform}
Any choice of $g(\cdot)$ satisfying the above conditions yields a lower bound on $\dot{C}_{\frac{1}{4\WW}}(0)$. Here we shall find the waveform that gives rise to the largest lower bound. Thus, we wish to maximize the RHS of \eqref{eq:main_LBgeneral} over all unit-energy waveforms $g(\cdot)$ that are bandlimited to $\WW$ Hz and that satisfy \eqref{eq:g1}--\eqref{eq:g3}, namely,
\begin{IEEEeqnarray*}{lCl}
\sum_{\ell\neq 0} \left|g\left(\frac{\ell-1/2}{2\WW}\right)\right| & < & \infty \\
\sum_{\ell\neq 0} \left|g\left(\frac{\ell}{2\WW}\right)\right| & < & \infty \\
\sum_{\ell\neq 0} \left|g\left(\frac{\ell+1/2}{2\WW}\right)\right| & < & \infty.
\end{IEEEeqnarray*}
To further facilitate the optimization problem, we choose $g(\cdot)$ to be a symmetric function, which implies that $\alpha_0=\gamma_0$. (Numerical computation suggests that this choice is indeed optimal.) In this case, the rate per unit-cost $\dot{R}(0)$ is of the form
\begin{equation*}
a\alpha_0^2 + 2b\alpha_0\beta_0 + c\beta_0^2
\end{equation*}
(where $a>0$, $c>0$ and $ac-b^2>0$), so $\dot{R}(0)$ is a convex function of $(\alpha_0,\beta_0)$, and all pairs $(\alpha_0,\beta_0)$ that give rise to the same rate per unit-cost lie on an ellipse. Let $\set{G}$ denote the set of all waveforms satisfying the above conditions, and let $\set{B}$ denote the set of all pairs $(\alpha_0,\beta_0)$ that arise from these waveforms. Since $\dot{R}(0)$ is convex in $(\alpha_0,\beta_0)$, it follows that the supremum of $\dot{R}(0)$ over $\set{B}$ lies on the boundary of $\set{B}$, which we shall determine next.

While the set of all bandlimited functions that satisfy \eqref{eq:g1}--\eqref{eq:g3} is convex, the set of all unit-energy functions is not (see \cite[Sec.~2]{rockafellar70} for a definition of a convex set). Indeed, if for example the functions $g_1(\cdot)$ and $g_2(\cdot)$ are of unit energy and satisfy $g_1(t)=-g_2(t)$, $t\in\Reals$, then
\begin{equation*}
t\mapsto\xi g_1(t)+(1-\xi)g_2(t),\qquad \xi\in[0,1]
\end{equation*}
 is of energy $(2\xi-1)^2$, which is not equal to $1$ unless $\xi=1$ or $\xi=0$. Nevertheless, the set of all functions whose energy is less than or equal to $1$ is convex, since
 \begin{equation*}
 \norm{\xi \vect{g}_1+(1-\xi)\vect{g}_2}^2 \leq \bigl(\xi\norm{\vect{g}_1}+(1-\xi)\norm{\vect{g}_2}\bigr)^2 \leq 1, \qquad \bigl(\norm{\vect{g}_1}^2\leq 1, \norm{\vect{g}_2}^2\leq 1, \xi\in[0,1]\bigr)
 \end{equation*}
where $\|\vect{g}\|\triangleq\sqrt{\int \bigl(g(t)\bigr)^2\d t}$, and where the inequality follows from the Triangle Inequality. Let $\set{G}'\supseteq \set{G}$ denote the set of all waveforms that are bandlimited to $\WW$ Hz, that satisfy \eqref{eq:g1}--\eqref{eq:g3}, and whose energy is less than or equal to $1$. Furthermore, let $\set{B}'\supseteq\set{B}$ denote the set of all pairs $(\alpha_0,\beta_0)$ that arise from these waveforms. Since convexity is preserved under intersection \cite[Thm.~2.1]{rockafellar70}, it follows that the set $\set{G}'$ (and hence also $\set{B}'$) is convex. We further have that
\begin{equation*}
\sup_{\set{G}} \dot{R}(0) = \sup_{\set{G}'} \dot{R}(0) = \sup_{\set{B}'}\dot{R}(0)
\end{equation*}
where the first step follows because every rate per unit-cost $\dot{R}(0)$ that corresponds to a waveform satisfying $\|\vect{g}\|<1$ can be increased by normalizing the waveform; and where the second step follows because $\dot{R}(0)$ depends on $g(\cdot)$ only via $(\alpha_0,\beta_0)$ and because $\set{G}'$ determines $\set{B}'$. The above optimization problem can thus be expressed as the maximization of a convex function over a convex set.

Let $\overline{\set{B}'}$ denote the closure of $\set{B}'$ (i.e., let $\overline{\set{B}'}$ be the smallest closed set containing $\set{B}'$) and let $\lambda\mapsto\zeta(\lambda)$ be defined as
\begin{equation}
\label{eq:main_support}
\zeta(\lambda) \triangleq \sup_{(\alpha_0,\beta_0)\in\set{B}'} \bigl| \lambda\,\alpha_0+\beta_0\bigr|, \qquad \lambda\in\Reals.
\end{equation}
Since $\overline{\set{B}'}$ is a closed convex set, it follows that the boundary of $\set{B}'$ is given by the points $(\alpha_0,\beta_0)\in\overline{\set{B}'}$ that achieve the supremum in \eqref{eq:main_support} \cite[Sec.~13]{rockafellar70}. Note that the boundary of $\set{B}'$ is symmetric with respect to the origin, since $(\alpha_0,\beta_0)$ and $(-\alpha_0,-\beta_0)$ yield the same value for $\bigl|\lambda\,\alpha_0+\beta_0\bigr|$, $\lambda\in\Reals$.

For every $\lambda\in\Reals$, the supremum in \eqref{eq:main_support} is achieved by the pair $(\alpha_0,\beta_0)$ that corresponds to $g(\cdot)$ being of Fourier Transform
\begin{equation}
\label{eq:main_spectrum}
\hat{g}(f) = \frac{1}{\sqrt{2\WW}} \frac{1+\lambda\cos\left(\pi\frac{f}{2\WW}\right)}{\sqrt{\frac{1}{2}\lambda^2+\frac{4}{\pi}\lambda+1}}\I{|f|\leq\WW}, \qquad f\in\Reals
\end{equation}
where $\I{\cdot}$ denotes the indicator function, i.e., $\I{\textnormal{statement}}$ is $1$ if the statement is true and $0$ otherwise. This yields
\begin{equation}
\label{eq:main_gopt1}
\alpha_0 = \frac{1}{\sqrt{\WW\Nzero}}\frac{\frac{2}{\pi}+\frac{1}{2}\lambda}{\sqrt{\frac{1}{2}\lambda^2+\frac{4}{\pi}\lambda+1}},\qquad \lambda\in\Reals
\end{equation}
and
\begin{equation}
\label{eq:main_gopt2}
\beta_0 = \frac{1}{\sqrt{\WW\Nzero}}\frac{1+\frac{2}{\pi}\lambda}{\sqrt{\frac{1}{2}\lambda^2+\frac{4}{\pi}\lambda+1}}, \qquad \lambda\in\Reals.
\end{equation}
Note that the waveform given by \eqref{eq:main_spectrum} is of unit energy and is bandlimited to $\WW$ Hz, but it does not satisfy \eqref{eq:g1}--\eqref{eq:g3}. Nevertheless, it is shown in Appendix~\ref{app:mainproof} that there exist pairs $(\alpha_0,\beta_0)\in\set{B}'$ that are arbitrarily close to \eqref{eq:main_gopt1} and \eqref{eq:main_gopt2}. The pairs $(\alpha_0,\beta_0)$ given by \eqref{eq:main_gopt1} and \eqref{eq:main_gopt2} are thus in the closure of $\set{B}'$ and the boundary of $\set{B}'$ is hence parameterized by \eqref{eq:main_gopt1} and \eqref{eq:main_gopt2}. (Notice that \eqref{eq:main_gopt1} and \eqref{eq:main_gopt2} describe only one half of the boundary of $\set{B}'$. The other half is given by $(-\alpha_0,-\beta_0)$.)

\begin{figure}
\centering
\psfrag{feasible region}[lt][lt]{\footnotesize boundary of $\set{B}'$}
\psfrag{I = 0.747}[cb][cb]{\footnotesize $\dot{R}(0)\approx 0.747$}
\psfrag{lambda=1.4}[cb][cb]{$\lambda=1.4$}
\psfrag{alpha}[cc][cc]{$\alpha_0$}
\psfrag{beta}[cc][cc]{$\beta_0$}
\epsfig{file=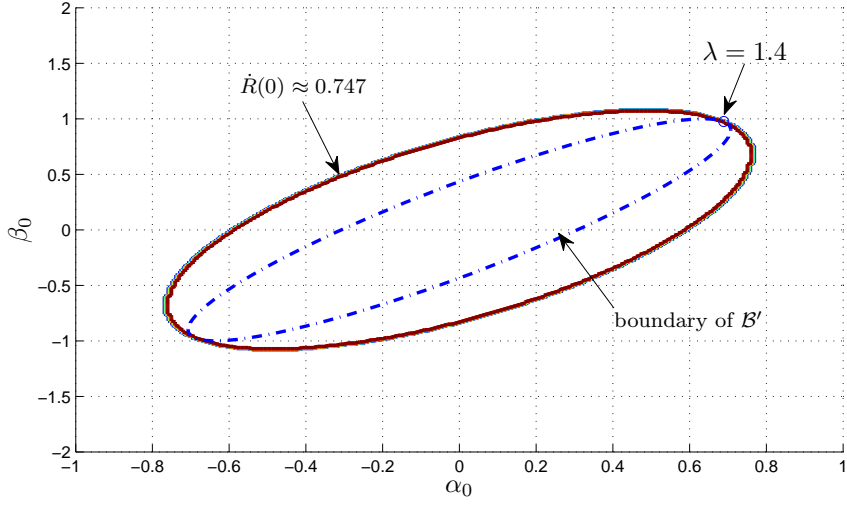, width=0.9\textwidth}
 \caption{The contour line corresponding to $\dot{R}(0)\approx 0.747$; the boundary of $\set{B}'$ as parameterized by \eqref{eq:main_gopt1} and \eqref{eq:main_gopt2}; and the point $(\alpha_0,\beta_0)$ that corresponds to $\lambda=1.4$.}
\label{fig2}
\end{figure}

The above problem is illustrated in Figure~\ref{fig2} for $\WW=\Nzero=1$. The outer curve (solid line) is the contour line corresponding to $\dot{R}(0)\approx 0.747$. The inner curve (dot-dashed line) depicts the boundary of $\set{B}'$ as parameterized by \eqref{eq:main_gopt1} and \eqref{eq:main_gopt2}. The two curves touch at two points: one corresponds to $\lambda=1.4$ and the other is the same point reflected through the origin. Since $\dot{R}(0)$ is convex in $(\alpha_0,\beta_0)$, it follows that any point that lies inside the depicted contour line (solid line) yields a rate per unit-cost that is smaller than $0.747$. Thus, the two points on the boundary of $\set{B}'$ that touch the contour line yield $\dot{R}(0)\approx 0.474$, whereas the other points $(\alpha_0,\beta_0)\in\set{B}'$ yield a rate per unit-cost that is smaller. Therefore, we conclude that choosing $(\alpha_0,\beta_0)$ according to \eqref{eq:main_gopt1} and \eqref{eq:main_gopt2} with $\lambda=1.4$ maximizes $\dot{R}(0)$.

In the following, we summarize the above arguments to compute the supremum of $\dot{R}(0)$ over the set $\set{B}'$ for general $\WW$ and $\Nzero$. To this end, we first recall that $\dot{R}(0)$ is convex in $(\alpha_0,\beta_0)$ and it therefore suffices to maximize $\dot{R}(0)$ over all boundary points of $\set{B}'$. We next use that every boundary point of $\set{B}'$ is given by the parametric equations \eqref{eq:main_gopt1} and \eqref{eq:main_gopt2}. Therefore, by applying \eqref{eq:main_gopt1} and \eqref{eq:main_gopt2} to \eqref{eq:main_LBgeneral}, the maximization of $\dot{R}(0)$ over the set $\set{B}'$ can be expressed as the maximization over the real scalar $\lambda$
\begin{IEEEeqnarray}{lCl}
\dot{C}_{\frac{1}{4\WW}}(0) & \geq &  \sup_{\lambda\in\Reals}\frac{2}{\pi}\frac{1}{\Nzero} \left[\frac{\left(\frac{\tilde{\alpha}(\lambda)}{2}+\frac{\tilde{\beta}(\lambda)}{4} +\frac{\tilde{\alpha}(\lambda)}{\pi}\arcsin\left(\frac{\rho}{\sqrt{1-\rho^2}}\right)-\frac{\tilde{\beta}(\lambda)}{2\pi}\arcsin\left(\frac{\rho^2}{1-\rho^2}\right)\right)^2}{\frac{1}{4} + \frac{1}{\pi} \arcsin(\rho)} \right.\nonumber\\
& &  \quad\qquad\qquad {} + 8\,\Biggl(\frac{\tilde{\beta}(\lambda)}{4}  -\frac{\tilde{\alpha}(\lambda)}{\pi}\arcsin\left(\frac{\rho}{\sqrt{1-\rho^2}}\right)+\frac{\tilde{\beta}(\lambda)}{2\pi}\arcsin\left(\frac{\rho^2}{1-\rho^2}\right)\Biggr)^2 \nonumber\\
& &  \quad\qquad\qquad {} +  \left.\frac{\left(\frac{\tilde{\alpha}(\lambda)}{2}-\frac{\tilde{\beta}(\lambda)}{4}-\frac{\tilde{\alpha}(\lambda)}{\pi}\arcsin\left(\frac{\rho}{\sqrt{1-\rho^2}}\right)+\frac{\tilde{\beta}(\lambda)}{2\pi}\arcsin\left(\frac{\rho^2}{1-\rho^2}\right)\right)^2}{\frac{1}{4}-\frac{1}{\pi}\arcsin(\rho)}\quad\,\,\right]\IEEEeqnarraynumspace\label{eq:main_beforeend}
\end{IEEEeqnarray}
where
\begin{equation*}
\tilde{\alpha}(\lambda) \triangleq \frac{\frac{2}{\pi}+\frac{1}{2}\lambda}{\sqrt{\frac{1}{2}\lambda^2+\frac{4}{\pi}\lambda+1}}, \qquad \lambda\in\Reals 
\end{equation*}
and
\begin{equation*}
\tilde{\beta}(\lambda) \triangleq  \frac{1+\frac{2}{\pi}\lambda}{\sqrt{\frac{1}{2}\lambda^2+\frac{4}{\pi}\lambda+1}},\qquad\lambda\in\Reals. 
\end{equation*}
Numerical computation shows that the supremum on the RHS of \eqref{eq:main_beforeend} is achieved for $\lambda=1.4$, which is consistent with our interpretation of Figure~\ref{fig2}. The capacity per unit-cost $\dot{C}_{\frac{1}{4\WW}}(0)$ is thus lower bounded
\begin{IEEEeqnarray}{lCl}
\dot{C}(0) & \geq & \left[\frac{\left(\frac{g_1}{2}+\frac{g_0}{4} +\frac{g_1}{\pi}\arcsin\left(\frac{\rho}{\sqrt{1-\rho^2}}\right)-\frac{g_0}{2\pi}\arcsin\left(\frac{\rho^2}{1-\rho^2}\right)\right)^2}{\frac{1}{4} + \frac{1}{\pi} \arcsin(\rho)} \right.\nonumber\\
& &  \quad {} + 8\,\Biggl(\frac{g_0}{4}  -\frac{g_1}{\pi}\arcsin\left(\frac{\rho}{\sqrt{1-\rho^2}}\right)+\frac{g_0}{2\pi}\arcsin\left(\frac{\rho^2}{1-\rho^2}\right)\Biggr)^2 \nonumber\\
& &  \quad {} +  \left.\frac{\left(\frac{g_1}{2}-\frac{g_0}{4}-\frac{g_1}{\pi}\arcsin\left(\frac{\rho}{\sqrt{1-\rho^2}}\right)+\frac{g_0}{2\pi}\arcsin\left(\frac{\rho^2}{1-\rho^2}\right)\right)^2}{\frac{1}{4}-\frac{1}{\pi}\arcsin(\rho)}\quad\,\,\right]\nonumber\\
& \approx & 0.747\frac{1}{\Nzero}
\end{IEEEeqnarray}
where
\begin{equation*}
g_1 = \left.\frac{\frac{2}{\pi}+\frac{1}{2}\lambda}{\sqrt{\frac{1}{2}\lambda^2+\frac{4}{\pi}\lambda+1}}\right|_{\lambda=1.4} \qquad \textnormal{and} \qquad g_0  \left.\frac{1+\frac{2}{\pi}\lambda}{\sqrt{\frac{1}{2}\lambda^2+\frac{4}{\pi}\lambda+1}}\right|_{\lambda=1.4}.
\end{equation*}
This proves Theorem~\ref{thm:main}.

\section{Summary and Conclusion}
\label{sec:summary}
We demonstrated that doubling the sampling rate recovers some of the loss in capacity per unit-cost incurred on the bandlimited Gaussian channel with a one-bit output quantizer. Indeed, when the channel output is sampled at Nyquist rate $2\WW$, it is well-known that the capacity per unit-cost is given by $\frac{2}{\pi}\frac{1}{\Nzero}\approx 0.64\frac{1}{\Nzero}$ \cite{viterbiomura79}, which is a factor of $\frac{2}{\pi}$ smaller than the capacity per unit-cost of the same channel but without output quantizer. We showed that, by sampling the output at twice the Nyquist rate, a capacity per unit-cost not less than $0.75\frac{1}{\Nzero}$ can be achieved. This can be viewed as a very-noisy counterpart of the work by Gilbert \cite{gilbert93} and by Shamai \cite{shamai94}, which demonstrated that oversampling increases the capacity of the above channel when there is no additive noise.

The conclusions that can be drawn from this result are twofold. Firstly, we demonstrated that in order to reduce the loss in capacity per unit-cost caused by the quantization, one can either increase the quantization resolution or the sampling rate. Thus, it is possible to trade amplitude resolution (quantization) versus time resolution (sampling rate). Secondly, we observe that while sampling the output at Nyquist rate is optimal for the AWGN channel (without output quantization), this does not hold when the output is quantized. Thus, a communication scheme that is optimal in the sense that it achieves the capacity need not be optimal anymore if the channel output is processed by a noninvertible operation (such as quantization).

\section*{Acknowledgment}
Fruitful discussions with Mich\`ele Wigger are gratefully acknowledged.

\appendix

\section{Appendix to Section~\ref{sub:compl}}
\label{app:compl}

\subsection{Bivariate Case}
\label{app:bivariate}
We show that, for $\const{A}\geq 0$, $\alpha\in\Reals$, and $\beta\in\Reals$, we have
\begin{equation*}
|\Delta_1(\const{A},\alpha,\beta)| \leq \const{A}^2\eta_1(\const{A},\alpha,\beta) \qquad \textnormal{and} \qquad
|\Delta_2(\const{A},\alpha,\beta)| \leq \const{A}^2\eta_2(\const{A},\alpha,\beta) 
\end{equation*}
where $\Delta_1(\const{A},\alpha,\beta)$ and $\Delta_2(\const{A},\alpha,\beta)$ are defined in \eqref{eq:binary_Delta1} and \eqref{eq:binary_Delta2}, and where $\eta_1(\const{A},\alpha,\beta)=\eta_1(\const{A},|\alpha|,|\beta|)$ and $\eta_2(\const{A},\alpha,\beta)=\eta_2(\const{A},|\alpha|,|\beta|)$ are monotonically increasing in $(\const{A},|\alpha|,|\beta|)$ and are bounded for every finite $\const{A}$, $\alpha$, and $\beta$. We have
\begin{IEEEeqnarray}{lCl}
\IEEEeqnarraymulticol{3}{l}{|\Delta_1(\const{A},\alpha,\beta)|}\nonumber\\
\quad  & \leq & \frac{|\beta|\const{A}}{\sqrt{2\pi(1-\varrho^2)}} \left|\int_{-\alpha\const{A}}^0 \phi_{0,1}(x)\d x\right| + \frac{|\varrho|}{\sqrt{2\pi(1-\varrho^2)}}\left|\int_{-\alpha\const{A}}^0 \phi_{0,1}(x)|x|\d x\right|\nonumber\\
& & {}  + \left|\int_{-\alpha\const{A}}^0 \phi_{0,1}(x)\left|\delta\left(\frac{\beta\const{A}+\rho x}{\sqrt{1-\varrho^2}}\right)\right|\d x\right| + \frac{1}{2}\bigl|\delta(\alpha\const{A})\bigr| \nonumber\\
& \leq & \frac{|\alpha|\,|\beta|}{\sqrt{2\pi(1-\varrho^2)}}\const{A}^2+\frac{|\varrho|}{\sqrt{2\pi(1-\varrho^2)}}\int^{|\alpha|\const{A}}_0 x\d x + \left|\int^{-\alpha\const{A}}_0 \frac{|\beta\const{A}+\rho x|^3}{6\sqrt{2\pi}(1-\varrho^2)^{\frac{3}{2}}}\d x\right| + \frac{|\alpha|^3\const{A}^3}{12\sqrt{2\pi}} \nonumber\\
& \leq & \frac{|\alpha|\,|\beta|}{\sqrt{2\pi(1-\varrho^2)}}\const{A}^2+\frac{|\varrho|\alpha^2}{\sqrt{2\pi(1-\varrho^2)}}\const{A}^2   + \frac{|\alpha|\bigl(|\beta|+|\rho|\,|\alpha|\bigr)^3}{6\sqrt{2\pi}(1-\varrho^2)^{\frac{3}{2}}}\const{A}^4 + \frac{|\alpha|^3}{12\sqrt{2\pi}}\const{A}^3 \nonumber\\
& = & \const{A}^2 \eta_1(\const{A},\alpha,\beta)
\end{IEEEeqnarray}
where
\begin{equation*}
\eta_1(\const{A},\alpha,\beta) \triangleq \frac{|\alpha|\,|\beta|}{\sqrt{2\pi(1-\varrho^2)}} + \frac{|\varrho|\alpha^2}{\sqrt{2\pi(1-\varrho^2)}} + \frac{|\alpha|\bigl(|\beta|+|\rho|\,|\alpha|\bigr)^3}{6\sqrt{2\pi}(1-\varrho^2)^{\frac{3}{2}}}\const{A}^2 + \frac{|\alpha|^3}{12\sqrt{2\pi}}\const{A}.
\end{equation*}
Here the first step follows from the Triangle Inequality; the second step follows by upper bounding $\phi_{0,1}(x) \leq 1$, $x\in\Reals$ and from \eqref{eq:delta_binary}; and the third step follows because, over the range of integration, $x\leq|\alpha|\const{A}$ and $|\beta\const{A}+\varrho x|\leq \bigl(|\beta|+|\varrho|\,|\alpha|\bigr)\const{A}$.

Along the same lines, we obtain
\begin{IEEEeqnarray}{lCl}
\IEEEeqnarraymulticol{3}{l}{|\Delta_2(\const{A},\alpha,\beta)|}\nonumber\\
\quad & \leq & \frac{|\varrho|}{\sqrt{2\pi(1-\varrho^2)}}\left|\int_{-\beta\const{A}}^0 \phi_{0,1}(y) |y| \d y\right| + \left|\int_{-\beta\const{A}}^0 \phi_{0,1}(y) \left|\delta\left(\frac{\varrho y}{\sqrt{1-\varrho^2}}\right)\right|\d y\right| + \frac{1}{2}\bigl|\delta(\beta\const{A})\bigr| \nonumber\\
& \leq & \frac{|\varrho|}{\sqrt{2\pi(1-\varrho^2)}}\int^{|\beta|\const{A}}_0 y \d y + \int^{|\beta|\const{A}}_0 \frac{\bigl(|\varrho| y\bigr)^3}{6\sqrt{2\pi}(1-\varrho^2)^{\frac{3}{2}}}\d y + \frac{|\beta|^3}{12\sqrt{2\pi}}\const{A}^3\nonumber\\
& = & \const{A}^2 \eta_2(\const{A},\alpha,\beta)
\end{IEEEeqnarray}
where
\begin{equation*}
\eta_2(\const{A},\alpha,\beta) \triangleq \frac{|\varrho|\beta^2}{\sqrt{2\pi(1-\varrho^2)}} + \frac{|\varrho|^3\beta^4}{6\sqrt{2\pi}(1-\varrho^2)^{\frac{3}{2}}}\const{A}^2 + \frac{|\beta|^3}{12\sqrt{2\pi}}\const{A}.
\end{equation*}
This proves the claim.

\subsection{Trivariate Case}
\label{app:trivariate}
We show that, for $\const{A}>0$, $\alpha\in\Reals$, $\beta\in\Reals$, and $\gamma\in\Reals$, we have
\begin{IEEEeqnarray*}{lCl}
|\Delta_1(\const{A},\alpha,\beta,\gamma)| & \leq & \const{A}^2\eta_1(\const{A},\alpha,\beta,\gamma)\\
|\Delta_2(\const{A},\alpha,\beta,\gamma)| & \leq & \const{A}^2\eta_2(\const{A},\alpha,\beta,\gamma)
\end{IEEEeqnarray*}
and
\begin{IEEEeqnarray*}{lCl}
|\Delta_3(\const{A},\alpha,\beta,\gamma)| \leq \const{A}^2\eta_3(\const{A},\alpha,\beta,\gamma)
\end{IEEEeqnarray*}
where $\Delta_1(\const{A},\alpha,\beta,\gamma)$, $\Delta_2(\const{A},\alpha,\beta,\gamma)$, and $\Delta_3(\const{A},\alpha,\beta,\gamma)$ are defined in \eqref{eq:ternary_Delta1}, \eqref{eq:ternary_Delta2}, and \eqref{eq:ternary_Delta3}, and where $\eta_1(\const{A},\alpha,\beta,\gamma)=\eta_1(\const{A},|\alpha|,|\beta|,|\gamma|)$, $\eta_2(\const{A},\alpha,\beta,\gamma)=\eta_2(\const{A},|\alpha|,|\beta|,|\gamma|)$, and $\eta_3(\const{A},\alpha,\beta,\gamma)=\eta_3(\const{A},|\alpha|,|\beta|,|\gamma|)$ are monotonically increasing in $(\const{A},|\alpha|,|\beta|,|\gamma|)$ and are bounded for every finite $\const{A}$, $\alpha$, $\beta$, and $\gamma$. We have
\begin{IEEEeqnarray}{lCl}
|\Delta_1(\const{A},\alpha,\beta,\gamma)| & \leq & |\delta(\alpha\const{A})|\left|\frac{1}{4}+\frac{1}{2\pi}\arcsin\left(\frac{\varrho_{23}-\varrho_{12}\varrho_{13}}{\sqrt{(1-\varrho_{12}^2)(1-\varrho_{13}^2)}}\right)\right| \nonumber\\
& & {} + \left|\int_{-\alpha\const{A}}^0\phi_{0,1}(x) \frac{1}{2\sqrt{2\pi}}\left|\frac{\beta\const{A}+\varrho_{12}x}{\sqrt{1-\varrho_{12}^2}}+\frac{\gamma\const{A}+\varrho_{13}x}{\sqrt{1-\varrho_{13}^2}}\right|\d x\right| \nonumber\\
& & {} + \left|\int_{-\alpha\const{A}}^0\phi_{0,1}(x) \left|\Delta\left(\const{A},\frac{\beta+\varrho_{12}x/\const{A}}{\sqrt{1-\varrho_{12}^2}},\frac{\gamma+\varrho_{13}x/\const{A}}{\sqrt{1-\varrho_{13}^2}}\right)\right|\d x\right|\nonumber\\
& \leq & \frac{1}{2}|\delta(\alpha\const{A})| + \frac{1}{2\sqrt{2\pi}} \left|\int_{-\alpha\const{A}}^0 \frac{1}{2\sqrt{2\pi}}\left|\frac{\beta\const{A}+\varrho_{12}x}{\sqrt{1-\varrho_{12}^2}}+\frac{\gamma\const{A}+\varrho_{13}x}{\sqrt{1-\varrho_{13}^2}}\right|\d x\right|\nonumber\\
& & {} + \left|\int_{-\alpha\const{A}}^0 \left|\Delta\left(\const{A},\frac{\beta+\varrho_{12}x/\const{A}}{\sqrt{1-\varrho_{12}^2}},\frac{\gamma+\varrho_{13}x/\const{A}}{\sqrt{1-\varrho_{13}^2}}\right)\right|\d x\right|\nonumber\\
& \leq & \frac{|\alpha|^3}{12\sqrt{2\pi}}\const{A}^3 + \frac{1}{2\sqrt{2\pi}}\int_0^{|\alpha|\const{A}} \left(\frac{|\beta|\const{A}+|\varrho_{12}|x}{\sqrt{1-\varrho_{12}^2}}+\frac{|\gamma|\const{A}+|\varrho_{13}|x}{\sqrt{1-\varrho_{13}^2}}\right) \d x\nonumber\\
& & {} + \const{A}^2 \int_{0}^{|\alpha|\const{A}} \eta\left(\const{A},\frac{|\beta|+|\varrho_{12}|x/\const{A}}{\sqrt{1-\varrho_{12}^2}},\frac{|\gamma|+|\varrho_{13}|x/\const{A}}{\sqrt{1-\varrho_{13}^2}}\right) \d x\nonumber\\
& \leq & \frac{|\alpha|^3}{12\sqrt{2\pi}}\const{A}^3 + \frac{|\alpha|}{2\sqrt{2\pi}}\left(\frac{|\beta|+|\varrho_{12}|\,|\alpha|}{\sqrt{1-\varrho_{12}^2}}+\frac{|\gamma|+|\varrho_{13}|\,|\alpha|}{\sqrt{1-\varrho_{13}^2}}\right) \const{A}^2\nonumber\\
& & {} + |\alpha|\, \eta\left(\const{A},\frac{|\beta|+|\varrho_{12}|\,|\alpha|}{\sqrt{1-\varrho_{12}^2}},\frac{|\gamma|+|\varrho_{13}|\,|\alpha|}{\sqrt{1-\varrho_{13}^2}}\right) \const{A}^3\nonumber\\
& = & \const{A}^2 \eta_1(\const{A},\alpha,\beta,\gamma)
\end{IEEEeqnarray}
where $\eta(\const{A},\alpha,\beta)$ is as in Proposition~\ref{prop:binary}, and where
\begin{IEEEeqnarray}{lCl}
\eta_1(\const{A},\alpha,\beta,\gamma) & \triangleq &  \frac{|\alpha|^3}{12\sqrt{2\pi}}\const{A} + \frac{|\alpha|}{2\sqrt{2\pi}}\left(\frac{|\beta|+|\varrho_{12}|\,|\alpha|}{\sqrt{1-\varrho_{12}^2}}+\frac{|\gamma|+|\varrho_{13}|\,|\alpha|}{\sqrt{1-\varrho_{13}^2}}\right)\nonumber\\
&& {} + |\alpha|\,\eta\left(\const{A},\frac{|\beta|+|\varrho_{12}|\,|\alpha|}{\sqrt{1-\varrho_{12}^2}},\frac{|\gamma|+|\varrho_{13}|\,|\alpha|}{\sqrt{1-\varrho_{13}^2}}\right) \const{A}.\nonumber
\end{IEEEeqnarray}
Here the first step follows from the Triangle Inequality; the second step follows by upper bounding $\phi_{0,1}(x)\leq 1$, $x\in\Reals$ and $\arcsin(x)\leq \pi/2$, $|x|\leq1$; the third step follows again from the Triangle Inequality, from the upper bound
\begin{equation*}
|\Delta(\const{A},\alpha,\beta)| \leq \const{A}^2 \eta(\const{A},\alpha,\beta)
\end{equation*}
and from the monotonicity of $\eta(\const{A},\alpha,\beta)$ in $(\const{A},|\alpha|,|\beta|)$; and the fourth step follows because, over the range of integration, $x\leq|\alpha|\const{A}$.

Along the same lines, we obtain for $\Delta_2(\const{A},\alpha,\beta,\gamma)$
\begin{IEEEeqnarray}{lCl}
|\Delta_2(\const{A},\alpha,\beta,\gamma)| & \leq & |\delta(\beta\const{A})|\left|\frac{1}{4} + \frac{1}{2\pi}\arcsin\left(\frac{\varrho_{13}-\varrho_{12}\varrho_{23}}{\sqrt{(1-\varrho_{12}^2)(1-\varrho_{23}^2)}}\right)\right| \nonumber\\
& & {} +  \left|\int_{-\beta\const{A}}^0\phi_{0,1}(y) \frac{1}{2\sqrt{2\pi}}\left|\frac{\varrho_{12}y}{\sqrt{1-\varrho_{12}^2}}+\frac{\gamma\const{A}+\varrho_{23}y}{\sqrt{1-\varrho_{23}^2}}\right|\d y\right|\nonumber\\
& & {} + \left|\int_{-\beta\const{A}}^0\phi_{0,1}(y) \left|\Delta\left(\const{A},\frac{\varrho_{12}y/\const{A}}{\sqrt{1-\varrho_{12}^2}},\frac{\gamma+\varrho_{23}y/\const{A}}{\sqrt{1-\varrho_{23}^2}}\right)\right|\d y\right| \nonumber\\
& \leq & \frac{|\beta|^3}{12\sqrt{2\pi}}\const{A}^3 + \frac{|\beta|}{2\sqrt{2\pi}}\left(\frac{|\varrho_{12}|\,|\beta|}{\sqrt{1-\varrho_{12}^2}}+\frac{|\gamma|+|\varrho_{23}|\,|\beta|}{\sqrt{1-\varrho_{23}^2}}\right)\const{A}^2\nonumber\\
& & {} + |\beta|\,\eta\left(\const{A},\frac{|\varrho_{12}|\,|\beta|}{\sqrt{1-\varrho_{12}^2}},\frac{|\gamma|+|\varrho_{23}|\,|\beta|}{\sqrt{1-\varrho_{23}^2}}\right) \const{A}^3\nonumber\\
& = & \const{A}^2 \eta_2(\const{A},\alpha,\beta,\gamma)
\end{IEEEeqnarray}
where
\begin{IEEEeqnarray}{lCl}
\eta_2(\const{A},\alpha,\beta,\gamma) & \triangleq & \frac{|\beta|^3}{12\sqrt{2\pi}}\const{A} + \frac{|\beta|}{2\sqrt{2\pi}}\left(\frac{|\varrho_{12}|\,|\beta|}{\sqrt{1-\varrho_{12}^2}}+\frac{|\gamma|+|\varrho_{23}|\,|\beta|}{\sqrt{1-\varrho_{23}^2}}\right)\nonumber\\
& & {} + |\beta|\,\eta\left(\const{A},\frac{|\varrho_{12}|\,|\beta|}{\sqrt{1-\varrho_{12}^2}},\frac{|\gamma|+|\varrho_{23}|\,|\beta|}{\sqrt{1-\varrho_{23}^2}}\right) \const{A}.\nonumber
\end{IEEEeqnarray}
Finally, we obtain for $\Delta_3(\const{A},\alpha,\beta,\gamma)$
\begin{IEEEeqnarray}{lCl}
|\Delta_3(\const{A},\alpha,\beta,\gamma)| & \leq &  |\delta(\gamma\const{A})|\left|\frac{1}{4}+\frac{1}{2\pi}\arcsin\left(\frac{\varrho_{12}-\varrho_{13}\varrho_{23}}{\sqrt{(1-\varrho_{13}^2)(1-\varrho_{23}^2)}}\right)\right|\nonumber\\
& & {} +  \left|\int_{-\gamma\const{A}}^0 \phi_{0,1}(z) \frac{1}{2\sqrt{2\pi}}\left|\frac{\varrho_{13}z}{\sqrt{1-\varrho_{13}^2}}+\frac{\varrho_{23}z}{\sqrt{1-\varrho_{23}^2}}\right| \d z\right| \nonumber\\
& & {} +  \left|\int_{-\gamma\const{A}}^0 \phi_{0,1}(z)\left|\Delta\left(\const{A},\frac{\varrho_{13}z/\const{A}}{\sqrt{1-\varrho_{13}^2}},\frac{\varrho_{23}z/\const{A}}{\sqrt{1-\varrho_{23}^2}}\right)\right|\d z\right|\nonumber\\
& \leq & \frac{|\gamma|^3}{12\sqrt{2\pi}}\const{A}^3 + \frac{|\gamma|}{2\sqrt{2\pi}}\left(\frac{|\varrho_{13}|\,|\gamma|}{\sqrt{1-\varrho_{13}^2}}+\frac{|\varrho_{23}|\,|\gamma|}{\sqrt{1-\varrho_{23}^2}}\right) \const{A}^2\nonumber\\
& & {} + |\gamma|\,\eta\left(\const{A},\frac{|\varrho_{13}|\,|\gamma|}{\sqrt{1-\varrho_{13}^2}},\frac{|\varrho_{23}|\,|\gamma|}{\sqrt{1-\varrho_{23}^2}}\right) \const{A}^3\nonumber\\
& = & \const{A}^2 \eta_3(\const{A},\alpha,\beta,\gamma)
\end{IEEEeqnarray}
where
\begin{IEEEeqnarray}{lCl}
\eta_3(\const{A},\alpha,\beta,\gamma) & \triangleq & \frac{|\gamma|^3}{12\sqrt{2\pi}}\const{A} + \frac{\gamma^2}{2\sqrt{2\pi}}\left(\frac{|\varrho_{13}|}{\sqrt{1-\varrho_{13}^2}}+\frac{|\varrho_{23}|}{\sqrt{1-\varrho_{23}^2}}\right)\nonumber\\
& & {} + |\gamma|\,\eta\left(\const{A},\frac{|\varrho_{13}|\,|\gamma|}{\sqrt{1-\varrho_{13}^2}},\frac{|\varrho_{23}|\,|\gamma|}{\sqrt{1-\varrho_{23}^2}}\right) \const{A}.\nonumber
\end{IEEEeqnarray}
This proves the claim.

\section{Appendix to Section~\ref{sub:mainproof}}
\label{app:mainproof}
We show that there exist bandlimited, unit-energy waveforms $g(\cdot)$ that satisfy \eqref{eq:g1}--\eqref{eq:g3} and that give rise to pairs $(\alpha_0,\beta_0)$ that are arbitrarily close to
\begin{equation}
\label{eq:app_alpha0}
\frac{1}{\sqrt{\WW\Nzero}}\frac{\frac{2}{\pi}+\frac{1}{2}\lambda}{\sqrt{\frac{1}{2}\lambda^2+\frac{4}{\pi}\lambda+1}},\qquad \lambda\in\Reals
\end{equation}
and
\begin{equation}
\label{eq:app_beta0}
\frac{1}{\sqrt{\WW\Nzero}}\frac{1+\lambda\frac{2}{\pi}}{\sqrt{\frac{1}{2}\lambda^2+\frac{4}{\pi}\lambda+1}}, \qquad \lambda\in\Reals.
\end{equation}
To this end, we first note that a sufficient condition for the waveform $g(\cdot)$ to satisfy \eqref{eq:g1}--\eqref{eq:g3} is
\begin{equation}
\label{eq:appmain_sufficient}
\bigl|g(t)\bigr| \leq \const{K}_1, \quad t\in\Reals \qquad \textnormal{and} \qquad |g(t)| \leq \frac{\const{K}_2}{t^2}, \quad |t|>\const{T}
\end{equation}
for some nonnegative $\const{K}_1$, $\const{K}_2$, and some positive $\const{T}$. Indeed, if \eqref{eq:appmain_sufficient} holds, then we have for every $\tau\in\Reals$
\begin{IEEEeqnarray}{lCl}
\sum_{\ell\neq 0} \left|g\left(\frac{\ell+\tau}{2\WW}\right)\right| & \leq & \sum_{\ell=-\infty}^{\infty} \left|g\left(\frac{\ell+\tau}{2\WW}\right)\right|\nonumber\\
& = & \sum_{\left|\frac{\ell+\tau}{2\WW}\right|\leq\const{T}+\frac{1}{2\WW}}\left|g\left(\frac{\ell+\tau}{2\WW}\right)\right| + \sum_{\left|\frac{\ell+\tau}{2\WW}\right|>\const{T}+\frac{1}{2\WW}}\left|g\left(\frac{\ell+\tau}{2\WW}\right)\right|\nonumber\\
& \leq & 2\,\const{K}_1\left(2\WW\,\const{T}+1\right) +  \sum_{\left|\frac{\ell+\tau}{2\WW}\right|>\const{T}+\frac{1}{2\WW}} \frac{\const{K}_2\bigl(2\WW\bigr)^2}{\left(\ell+\tau\right)^2}\nonumber\\
& \leq & 2\,\const{K}_1 \left(2\WW\,\const{T}+1\right) + 2\,\int_{2\WW\const{T}+1}^{\infty} \frac{\const{K}_2\bigl(2\WW\bigr)^2}{\left(t-1\right)^2} \d t\nonumber\\
& = & 2\,\const{K}_1\left(2\WW\,\const{T}+1\right) + 2\,\const{K}_2 \frac{2\WW}{\const{T}}\nonumber\\
& < & \infty
\end{IEEEeqnarray}
where the first step follows because $\bigl|g\bigl(\tau/(2\WW)\bigr)\bigr|$ is nonnegative; the third step follows from \eqref{eq:appmain_sufficient} and because there are not more than $2\WW\,\const{T}+1$ terms satisfying $\bigl|(\ell+\tau)/(2\WW)\bigr|\leq \const{T}+1/(2\WW)$; and the fourth step follows by upper bounding the sum by an integral.

The waveform $g(\cdot)$ corresponding to the Fourier Transform \eqref{eq:main_spectrum} is given by
\begin{equation}
 g(t) = \sqrt{\frac{2\WW}{\frac{1}{2}\lambda^2+\frac{4}{\pi}\lambda+1}}\Biggl(\sinc(2\WW\,t)+\frac{\lambda}{2}\sinc\left(2\WW\,t-\frac{1}{2}\right)+\frac{\lambda}{2}\sinc\left(2\WW\,t+\frac{1}{2}\right)\Biggr), \quad t\in\Reals.\label{eq:appmain_g}
\end{equation}
We note that the sum of the last two terms on the RHS of \eqref{eq:appmain_g} equals
\begin{equation*}
\frac{\lambda}{2}\sinc\left(2\WW\,t-\frac{1}{2}\right)+\frac{\lambda}{2}\sinc\left(2\WW\,t+\frac{1}{2}\right) = \left\{\begin{array}{cc} \displaystyle\frac{\lambda}{2}, \quad & t=\pm\frac{1}{4\WW} \\[8pt] \displaystyle -\frac{\lambda}{2}\frac{\cos(2\pi\WW\,t)}{\pi\bigl((2\WW\,t)^2-\frac{1}{4}\bigr)}, \quad & t\neq \pm\frac{1}{4\WW}\end{array} \right.
\end{equation*}
which satisfies \eqref{eq:appmain_sufficient} and hence also \eqref{eq:g1}--\eqref{eq:g3}. However, $t\mapsto\sinc(2\WW\,t)$ decays like $1/t$ and does neither satisfy \eqref{eq:g1} nor \eqref{eq:g3}. We therefore replace $t\mapsto\sinc(2\WW\,t)$ in \eqref{eq:appmain_g} by a raised-cosine pulse of bandwidth $\WW$ Hz and of roll-off factor $\xi$ to obtain the desired waveform
\begin{IEEEeqnarray}{lCl}
g_{\xi}(t) & = & \sqrt{\frac{2\WW}{\Psi(\xi)}} \left(\frac{1}{1+\xi}\sinc\left(\frac{2\WW\,t}{1+\xi}\right)\frac{\cos\left(\pi2\WW\frac{\xi}{1+\xi}t\right)}{1-4\left(2\WW\frac{\xi}{1+\xi}t\right)^2}\right.\nonumber\\
& & \qquad\qquad\quad {}  + \left.\vphantom{\frac{\cos\left(\pi2\WW\frac{\xi}{1+\xi}t\right)}{1-4\left(2\WW\frac{\xi}{1+\xi}t\right)^2}} \frac{\lambda}{2}\sinc\left(2\WW\,t-\frac{1}{2}\right)+\frac{\lambda}{2}\sinc\left(2\WW\,t+\frac{1}{2}\right)\right), \qquad t\in\Reals\label{eq:appmain_desired}
\end{IEEEeqnarray}
where
\begin{IEEEeqnarray}{lCl}
\Psi(\xi) & \triangleq & 2\WW \int \left(\frac{1}{1+\xi}\sinc\left(\frac{2\WW\,t}{1+\xi}\right)\frac{\cos\left(\pi2\WW\frac{\xi}{1+\xi}t\right)}{1-4\left(2\WW\frac{\xi}{1+\xi}t\right)^2}\right.\nonumber\\
& & \qquad\qquad {} +\left.\vphantom{\frac{\cos\left(\pi2\WW\frac{\xi}{1+\xi}t\right)}{1-4\left(2\WW\frac{\xi}{1+\xi}t\right)^2}}\frac{\lambda}{2}\sinc\left(2\WW\,t-\frac{1}{2}\right)+\frac{\lambda}{2}\sinc\left(2\WW\,t+\frac{1}{2}\right)\right)^2 \d t.\nonumber
\end{IEEEeqnarray}
For every $\xi>0$, the raised-cosine pulse on the RHS of \eqref{eq:appmain_desired} satisfies \eqref{eq:appmain_sufficient} and hence also \eqref{eq:g1}--\eqref{eq:g3}. It therefore follows from the Triangle Inequality that also the pulse $g_{\xi}(\cdot)$ satisfies \eqref{eq:g1}--\eqref{eq:g3}, since for every $t\in\Reals$
\begin{IEEEeqnarray}{lCl}
\left|g_{\xi}\left(t\right)\right| & \leq & \sqrt{\frac{2\WW}{\Psi(\xi)}}\left|\frac{1}{1+\xi}\sinc\left(\frac{2\WW\,t}{1+\xi}\right)\frac{\cos\left(\pi2\WW\frac{\xi}{1+\xi}t\right)}{1-4\left(2\WW\frac{\xi}{1+\xi}t\right)^2}\right|\nonumber\\
& & {} + \sqrt{\frac{2\WW}{\Psi(\xi)}}\left|\frac{\lambda}{2}\sinc\left(2\WW\,t-\frac{1}{2}\right)+\frac{\lambda}{2}\sinc\left(2\WW\,t+\frac{1}{2}\right)\right|.\nonumber
\end{IEEEeqnarray}
Furthermore, since $g_{\xi}(\cdot)$ is of unit energy and bandlimited to $\WW$ Hz, it follows that $g_{\xi}(\cdot)$ is in $\set{G}'$, and hence the pairs $(\alpha_0,\beta_0)$ that arise from $g_{\xi}(\cdot)$ are in $\set{B}'$.

It remains to show that such $(\alpha_0,\beta_0)$ approach \eqref{eq:app_alpha0} and \eqref{eq:app_beta0} as $\xi$ tends to zero. To this end, we show that for every $t\in\Reals$
 \begin{equation}
 \label{eq:appmain_limit}
\lim_{\xi\downarrow 0} g_{\xi}(t) = g(t).
\end{equation}
This implies
\begin{IEEEeqnarray}{lCcCl}
\lim_{\xi\downarrow 0}\alpha_0 & = & \lim_{\xi\downarrow 0} \frac{1}{\sqrt{(2\WW)(\WW\Nzero)}}g_{\xi}\left(-\frac{1}{4\WW}\right) & = & \frac{1}{\sqrt{\WW\Nzero}}\frac{\frac{2}{\pi}+\frac{1}{2}\lambda}{\sqrt{\frac{1}{2}\lambda^2+\frac{4}{\pi}\lambda+1}} \nonumber
\end{IEEEeqnarray}
and
\begin{IEEEeqnarray}{lCcCl}
\lim_{\xi\downarrow 0}\beta_0 & = & \lim_{\xi\downarrow 0} \frac{1}{\sqrt{(2\WW)(\WW\Nzero)}}g_{\xi}\left(0\right) & = & \frac{1}{\sqrt{\WW\Nzero}}\frac{1+\frac{2}{\pi}\lambda}{\sqrt{\frac{1}{2}\lambda^2+\frac{4}{\pi}\lambda+1}} \nonumber
\end{IEEEeqnarray}
which in turn proves the claim.

To prove \eqref{eq:appmain_limit}, we note that $t\mapsto \sqrt{\Psi(\xi)}g_{\xi}(t)$---which we shall denote by $t\mapsto h_{\xi}(t)$---satisfies
\begin{equation}
\lim_{\xi\downarrow 0} h_{\xi}(t) = \sqrt{\left(\frac{1}{2}\lambda^2+\frac{4}{\pi}\lambda+1\right)}\, g(t), \qquad t\in\Reals \label{eq:appmain_unnormalized}
\end{equation}
where $g(\cdot)$ is the Inverse Fourier Transform of \eqref{eq:main_spectrum}.
Furthermore, since the Fourier Transform of $h_{\xi}(\cdot)$ satisfies $|\hat{h}_{\xi}(f)|^2\leq\bigl(1/2\,\lambda^2+4/\pi\,\lambda+1\bigr)|\hat{g}(f)|^2$, $f\in\Reals$ (with $\hat{g}(\cdot)$ given by \eqref{eq:main_spectrum}), and since $\int\bigl(1/2\,\lambda^2+4/\pi\,\lambda+1\bigr)|\hat{g}(f)|^2\d f<\infty$, it follows from the Dominated Convergence Theorem \cite[Thm.~1.34]{rudin87} that
\begin{equation}
\lim_{\xi\downarrow 0} \Psi(\xi) = \lim_{\xi\downarrow 0} \int \bigl(h_{\xi}(t)\bigr)^2\d t = \lim_{\xi\downarrow 0}\int \left|\hat{h}_{\xi}(f)\right|^2\d f = \int \lim_{\xi\downarrow 0}\left|\hat{h}_{\xi}(f)\right|^2 \d f = \frac{1}{2}\lambda^2 + \frac{4}{\pi}\lambda + 1 \label{eq:appmain_Psilim}
\end{equation}
where the second step follows from Parseval's Theorem \cite[Thm.~6.2.9]{lapidoth09}, and where the last step follows from \eqref{eq:appmain_unnormalized} and because $g(\cdot)$ is of unit energy. It thus follows that for every $t\in\Reals$
\begin{equation}
\lim_{\xi\downarrow 0} g_{\xi}(t) = \lim_{\xi\downarrow 0}  \frac{h_{\xi}(t)}{\sqrt{\Psi(\xi)}} = \frac{\lim_{\xi\downarrow 0} h_{\xi}(t)}{\sqrt{\lim_{\xi\downarrow 0}\Psi(\xi)}}  = g(t)
\end{equation}
where the last step follows from \eqref{eq:appmain_unnormalized} and \eqref{eq:appmain_Psilim}. This proves \eqref{eq:appmain_limit} and hence also the claim.



\end{document}